\documentclass[11pt]{article}
\usepackage[english]{babel}
\usepackage[utf8]{inputenc}
\usepackage[T1]{fontenc}
\usepackage{amsmath,amssymb,lmodern,amsthm,accents}
\usepackage{marginnote}
\usepackage{mathtools}
\usepackage{graphicx}
\usepackage{mathrsfs}
\usepackage{bm}
\usepackage{url}
\usepackage{breakurl}
\usepackage{array}
\usepackage{tensor}
\usepackage{pifont}
\usepackage[toc,page]{appendix}
\usepackage[left=2.5cm,right=2.5cm,top=3cm,bottom=3cm]{geometry}
\usepackage{tikz}
\usetikzlibrary{arrows,positioning}
\usepackage{enumitem}
\usepackage{empheq}
\usepackage{cases}
\usepackage{csquotes}
\usepackage{setspace}

\theoremstyle{plain}
\usepackage{xargs}
\usepackage{xcolor}
\usepackage{subcaption}
\usepackage[justification=raggedright]{caption}
\usepackage{hyperref}
\usepackage{cleveref}
\usepackage{tocloft}
\hypersetup{
	colorlinks=true,
	urlcolor=blue,
	linkcolor={black},
	citecolor={green!40!black}
}
\usepackage[affil-it]{authblk}
\usepackage{thmtools}
\usepackage[backend=bibtex,style=numeric-comp,maxnames=4,sorting=none,url=false,hyperref=true,date=year,eprint=false,doi=true]{biblatex}

\newcommand{\ubar}[1]{\underaccent{\bar}{#1}}

\numberwithin{equation}{section}
\newtheorem{thm}{Theorem}[]
\theoremstyle{definition}

\newtheorem*{crit*}{Criterion}
\newtheorem{crit}{Criterion}[]

\newtheorem{lemma}{Lemma}[section]
\newtheorem{remark}{Remark}[section]

\newcommand{\defeq}{\vcentcolon=}
\newcommand{\defeqs}{\stackrel{\scri}{\vcentcolon=}}
\newcommand{\scri}{\mathscr{J}}
\newcommand{\ct}[2]{\tensor{{#1}}{#2}}
\newcommand{\cts}[2]{\tensor{\overline{#1}}{#2}}
\newcommand{\ctr}[3]{\tensor[{#1}]{#2}{#3}}

\newcommand{\cth}[2]{\tensor{\hat{#1}}{#2}}

\newcommand{\prn}[1]{\left(#1\right)}

\newcommand{\brkt}[1]{\left[#1\right]}

\newcommand{\cbrkt}[1]{\left\lbrace#1\right\rbrace}

\newcommand{\cd}[1]{\tensor{\nabla}{#1}}

\newcommand{\C}{\mathcal{C}}
\newcommand{\Q}{\mathcal{Q}}

\newcommand{\D}{\mathcal{D}}
\newcommand{\W}{\mathcal{W}}

\renewcommand{\P}{\mathcal{P}}

\newcommand{\B}{\mathcal{B}}

\newcolumntype{M}[1]{>{\centering\arraybackslash}m{#1}}
\newcolumntype{N}{@{}m{0pt}@{}}

\newcounter{marginnotecount}[section]

	\title{\Large \textbf{Gravitational radiation at infinity  with\\ negative cosmological constant and AdS$_4$ holography}}
	\author[]{Francisco Fernández-Álvarez\thanks{francisco.fernandez@ehu.eus}\ }
	\author[]{\ José M. M. Senovilla\thanks{josemm.senovilla@ehu.eus}} 
	\affil[]{Departamento de Física\\ Universidad del País Vasco UPV/EHU\\ Apartado 644, 48080 Bilbao, Spain }
	\date{\today{}}
\begin{document}
	
	\maketitle

	\begin{abstract}
	 The covariant characterization of the existence of gravitational radiation traversing infinity $\scri$ in the presence of a negative cosmological constant is presented. It is coherent and consistent with the previous characterizations put forward for the cases of non-negative cosmological constant, relying on the properties of the {\em asymptotic} super-Poynting vector; or in more transparent terms, based on the intrinsic properties of the flux of tidal energy at infinity. The proposed characterization is fully satisfactory, it can be  covariantly typified in terms of boundary data at infinity, and it can also be categorized according to the geometric properties of the rescaled Weyl tensor at $\scri$. The cases with no incoming radiation entering from (or no outgoing radiation escaping at) $\scri$ can similarly be  determined in terms of the boundary data or geometric properties of the rescaled Weyl tensor. In particular, we identify the most general boundary conditions that, in an initial-boundary value problem, ensure absence of gravitational radiation traversing $\scri$, namely (functional) proportionality between the Cotton-York tensor field and the holographic stress tensor field at $\scri$. We also present novel conditions ensuring the absence of just incoming (outgoing) radiation at $\scri$. These are given in a covariant way and also in terms of standard rescaled Weyl tensor scalars. The results are compatible with any matter content of the physical spacetime.
	\end{abstract}

	\tableofcontents

\section{Introduction}
In the last decade many efforts were devoted to understand how to characterize gravitational radiation at infinity in the presence of a positive cosmological constant $\Lambda$, see \cite{Ashtekar2014,Chrusciel2016,Ashtekar2017,Saw2018i,Fernandez-Alvarez-Senovilla2020b,Fernandez-Alvarez-Senovilla2022b} and references therein. Now, there is a well-established criterion, based on the existence of a flux of tidal energy at infinity, that identifies the presence of gravitational radiation there \cite{Fernandez-Alvarez-Senovilla2020b,Fernandez-Alvarez-Senovilla2022b,Senovilla2022} for $\Lambda\geq 0$. The criterion was proven to be equivalent to the traditional characterization using the news tensor for the case of asymptotically flat spacetimes (vanishing cosmological constant) \cite{Fernandez-Alvarez-Senovilla2020a,Fernandez-Alvarez-Senovilla2022a}. Furthermore, recent application to a very general family of exact solutions representing (pairs of) black holes with $\Lambda >0$ have shown that, according to the criterion, gravitational radiation escapes to infinity if and only if the black holes are accelerated \cite{Fernandez-Alvarez-Podolsky-Senovilla2024} ---a very satisfactory result.

The problem when the cosmological constant is negative has been much less explored. Notable exceptions are \cite{Compere2019}, where an analysis of boundary conditions that keep good properties concerning the Cauchy problem and the limit to vanishing $\Lambda$ was presented, \cite{Poole2019} with an extension of the standard Bondi expansion to the $\Lambda <0$ case, a study of the orientation of principal null directions at infinity \cite{KrtousPodolsky2004}, and the recent \cite{Ciambelli2024} where application of our tidal approach was tried with interesting consequences, despite the fact that actually only necessary conditions were imposed ---as we will explain later in remark \ref{sabrina}. 

The purpose of this paper is to apply the tidal approach to the case of asymptotics with $\Lambda <0$, to derive the necessary and sufficient conditions for the existence of gravitational radiation at infinity (in vacuum and in the presence of matter that decays asymptotically), to distinguish between incoming and outgoing cases, and to provide novel boundary conditions, or holographic data, with a clear interpretation in terms of radiation.

 The tidal approach uses superenergy methods ---see \cite{Senovilla2000} and references therein for full details---  and conformal completions à la Penrose \cite{Penrose63}. The approach relies on the idea that the Weyl curvature is the real gravitational field, being responsible  of the tidal distortions in spacetime. Hence, gravitational radiation transports tidal energy, ergo the existence of such radiation must be identified by a flux of tidal energy. The quantity that measures such a flux for a given observer  is called the super-Poynting vector ---in analogy with the electromagnetic Poynting vector ---, it was introduced in \cite{Bel1958,Bel1962}, and its relevance to characterize intrinsic gravitational radiation was quickly understood and explored long ago \cite{Bel1962,Pirani57,Wheeler77,Zakharov}, see also \cite{Senovilla2000,Alfonso2008}. For a given observer, the super-Poynting vector is given by the commutator of the electric and magnetic parts of the Weyl tensor relative to that observer \cite{Bel1962,Wheeler77,Zakharov,Senovilla2000,Alfonso2008,Fernandez-Alvarez-Senovilla2022b}. 
 
 The question is how to apply these ideas to infinity, as many subtleties arise. To fix the ideas, let us recall the concepts of conformal completion and conformal infinity \cite{Penrose63}. One works in an unphysical spacetime $ M $ with Lorentzian metric $\ct{g}{_{\alpha\beta}}$, that contains the physical spacetime $(\hat{M} , \cth{g}{_{\alpha\beta}})$ ---both of them are assumed to be connected and time oriented. The physical and unphysical metrics are related by a conformal factor $\Omega>0$ on $\hat{M}$:  $\ct{g}{_{\alpha\beta}}=\Omega^2\cth{g}{_{\alpha\beta}}$. Such a {\em conformal completion} $\prn{M,\ct{g}{_{\alpha\beta}}}$ is not unique, as there is the gauge freedom $\Omega \rightarrow \omega \Omega$ with $\omega$ a positive function on $M$; \emph{throughout this work, this gauge is not fixed in any way}. In this conformal picture, infinity is represented by the boundary of $\hat{M}$ in $M$, given by $\Omega =0$ and denoted by $\scri$, which is a hypersurface that, depending on the sign of the cosmological constant $\Lambda$, has a lightlike ($\Lambda$=0), spacelike ($\Lambda>0$) or timelike ($\Lambda<0$) character. The latter is the case that remains to be addressed, and the issue of this article.

 Importantly, the Weyl tensor $\ct{C}{_{\alpha\beta\gamma}^{\delta}}$ vanishes at infinity $\scri$ (i.e., at the conformal boundary) and thus all possible super-Poynting vectors vanish there too. This problem is overcome because one can use the rescaled Weyl tensor,
 	\begin{equation}
 		\ct{d}{_{\alpha\beta\gamma}^{\delta}}\defeq\frac{1}{\Omega}\ct{C}{_{\alpha\beta\gamma}^{\delta}}\ ,
 	\end{equation}
which is regular at $\scri$ and describes the {\em asymptotic} behaviour of the Weyl curvature ---see \cite{Fernandez-Alvarez-Senovilla2022a} for full details on the conventions adopted here. From the space-time viewpoint, this tensor contains the gravitational {\em asymptotic} information, in particular information about the gravitational radiation at infinity. Moreover, as discussed later and also in \cite{Fernandez-Alvarez-Senovilla2022a,Fernandez-Alvarez-Senovilla2022b,Senovilla2022}, such information must be encoded in initial/boundary data and the way to address this question varies depending on the causal character of $\scri$ ---that is, on the sign of $\Lambda$.

 Let $\ct{u}{^{\alpha}}$ denote a non-spacelike, future-pointing, vector field at $\scri$. Then, define the \emph{asymptotic supermomentum} relative to $\ct{u}{^{\alpha}}$ at $\scri$ as
 	\begin{equation}\label{supermomenta}
 		\ct{\P}{^\alpha}\prn{\vec{u}}\defeqs -\ct{u}{^{\rho}}\ct{u}{^{\sigma}}\ct{u}{^{\mu}}\ct{\D}{^{\alpha}_{\rho\sigma\mu}}\ ,
	 	\end{equation}
 where the rescaled Bel-Robinson tensor $\ct{\D}{_{\alpha\beta\gamma\delta}}$ is defined in terms of the rescaled Weyl tensor as
 	\begin{equation}
 			  \ct{\D}{_{\alpha\beta\gamma\delta}}\defeq \ct{d}{_{\alpha\mu\gamma}^\nu}\,\ct{d}{_{\delta\nu\beta}^\mu}
 	     +  \ctr{^*}{d}{_{\alpha\mu\gamma}^\nu}\,\ctr{^*}{d}{_{\delta\nu\beta}^\mu} =
	     \ct{d}{_{\alpha\mu\gamma}^\nu}\,\ct{d}{_{\delta\nu\beta}^\mu}+\ct{d}{_{\alpha\mu\delta}^\nu}\,\ct{d}{_{\gamma\nu\beta}^\mu}-\frac{1}{8}\ct{g}{_{\alpha\beta}}\ct{g}{_{\gamma\delta}} \ct{d}{_{\rho\mu\sigma}^\nu}\,\ct{d}{^{\rho\mu\sigma}_\nu} \, .
 			\end{equation}
			Here, the star operator is the Hodge dual defined by
			$$
			\ctr{^*}{d}{_{\alpha\mu\gamma}^\nu}:=\frac{1}{2}\eta_{\alpha\mu}{}^{\rho\sigma}  \ct{d}{_{\rho\sigma\gamma}^\nu}
			$$
			where $\eta_{\alpha\beta\lambda\mu}$ is the canonical volume element 4-form in $\prn{M,\ct{g}{_{\alpha\beta}}}$.
 The asymptotic supermomentum was introduced in \cite{Fernandez-Alvarez-Senovilla2020a, Fernandez-Alvarez-Senovilla2020b} and also used in \cite{Fernandez-Alvarez-Senovilla2022a,Fernandez-Alvarez-Senovilla2022b} ---see \cite{Senovilla2022} for a review--- for the cases with $\Lambda=0$ and $\Lambda>0$. In those situations, one uses the normal to $\scri$, $\ct{N}{_{\alpha}}\defeq \cd{_{\alpha}}\Omega$, as a canonical choice for $\ct{u}{_{\alpha}}$, as this is a geometrically selected, {\em causal} vector field when $\scri$ is lightlike or spacelike. When $\Lambda <0$, however,  there is not such a canonical choice, because $\scri$ is by itself a 3-dimensional Lorentzian manifold and there is no particular preferred choice for a causal vector field there ---the normal being spacelike. This is another issue that we must address, and the solution is rather simple: we stick to our fundamental idea, that is to say, ``natural observers'' at $\scri$ should not see any flux of tidal energy  if there is no gravitational radiation traversing $\scri$. This simple but powerful idea provides the right answer, which we present in the next section in the form of some criteria to detect the existence of gravitational radiation.
 
In later sections we will show the covariant formulation of the criteria, their relation with the geometry of the rescaled Weyl tensor at $\scri$,
 and the implications for boundary, or holographic, data in asymptotically AdS spacetimes. In particular, we present the most general holographic data compatible with the absence of radiation at $\scri$, given by the {\em functional proportionality} of the Cotton-York tensor field  $\ct{C}{_{\alpha\beta}}$ and the holographic stress tensor field $\ct{D}{_{\alpha\beta}}$ at $\scri$ (see \eqref{eq:n-electric} and \eqref{eq:n-magnetic} for precise definitions), namely
 	\begin{equation}\label{eq:condition-prop-covariant-intro}
 	 				\beta \ct{D}{_{\alpha\beta}}=\gamma \ct{C}{_{\alpha\beta}}\ , \hspace{5mm} \gamma^2+\beta^2 \neq 0\ ,
 		\end{equation}
 where $\gamma$ and $\beta$ are real functions on $\scri$. This criterion is independent of the matter content of the physical spacetime, and it works equally well in the presence or absence of matter, as long as the physical energy-momentum tensor satisfies the following decaying condition\footnote{This is necessary for the conformal completion à la Penrose to be well defined; see \cite{Fernandez-Alvarez-Senovilla2022a} for full details.}
 \begin{equation}\label{T->0}
 \lim_{\Omega \rightarrow 0} \frac{\hat{T}_{\mu\nu} }{\Omega} := \tau_{\mu\nu} \mbox{(finite)}.
\end{equation}
Note that the divergence of $\ct{C}{_{\alpha\beta}}$ vanishes, while the divergence of $\ct{D}{_{\alpha\beta}}$ depends on the leading term $\tau_{\mu\nu}$ of the energy-momentum tensor and, in general, does not vanish. Of course, if $\tau_{\mu\nu}=0$, in particular in vacuum, $D_{\alpha\beta}$ is divergence-free. Remarkably, other conditions presented in the literature \cite{deHaro2008,Mukhopadhyay2013,Ciambelli2018} are particular cases of \cref{eq:condition-prop-covariant-intro} for constant $\gamma$ and $\beta$, and thus for  divergence-free $\ct{D}{_{\alpha\beta}}$, as their goal was to consider vacuum asymptotically AdS spacetimes.
 
Moreover, we give novel holographic data capable of just removing incoming (outgoing) radiation at $\scri$ ---these criteria apply to both general and algebraically special asymptotic geometries, such as solutions studied in the context of the fluid/gravity correspondence \cite{BernardideFreitas2014,GathEtal2015}. The results are fully satisfactory and fit particularly well with the standard formulations of the initial-boundary value problem relevant to asymptotically AdS spacetimes.

\section{Characterisation of gravitational radiation at infinity}
 From now on, we assume that $\Lambda <0$, so that $\scri$ is a timelike hypersurface in the unphysical spacetime. Since the conformal factor $\Omega$ vanishes at $\scri$ and $\Lambda<0$, the unit normal to $\scri$ is a spacelike vector field given by
  	\begin{equation}
  		\ct{n}{_{\alpha}}\defeqs\sqrt{ -\frac{\Lambda}{3}}\ct{N}{_{\alpha}}\ ,
  	\end{equation}
 where $\ct{N}{_{\alpha}}\defeq \cd{_{\alpha}}\Omega$. With this definition, $\ct{n}{^{\alpha}}$ points `inwards' ---i.e., towards the physical spacetime.\\
 
In order to define an asymptotic super-momentum, one faces the problem of the absence of a canonical observer at $\scri$, due to the timelike character of $\scri$ and the spacelike character of $N^\alpha$. This is the most distinguishing feature of the present case with $\Lambda<0$.  There are plenty of observers at $\scri$, and they should be treated on an equal footing\footnote{For further arguments supporting this idea, see the discussion on initial data of \cref{sec:IBVP}.}. 
 Taking this democratic principle into account, and recalling that our fundamental idea is that radiation does not reach infinity if the {\em geometrically selected observers there} do not see any transverse flux of tidal energy, the gravitational radiation criterion can be formulated as follows:
 	\begin{crit}[No gravitational radiation condition with $\Lambda<0$] \label{crit:grav-rad}
	There is no gravitational radiation at $\scri$ if and only if there is no transverse flux of asymptotic superenergy \eqref{supermomenta} for observers within $\scri$. That is,
 	 		\begin{equation}\label{eq:condition}
 	 			\ct{n}{_{\mu}}\ct{\P}{^{\mu}}\prn{\vec{u}}=0\ \ \ \ \ \ \forall\ \ct{u}{^{\alpha}}, \ \ \  \ct{u}{^{\mu}}\ct{u}{_{\mu}}=-1\ , \ \ \ \ct{u}{^{\mu}}\ct{n}{_{\mu}}=0\ .
 	 		\end{equation}	 
 	\end{crit}
  \begin{remark}
  As usual, given the observer $u^\alpha$, one can decompose the corresponding supermomentum like\footnote{Later on we will use bars for complex conjugation, as is customary. The bar on top on the super-Poynting does not have this meaning. Notice the subtle point that this bar is larger.} 
  $$
  \ct{\P}{^{\mu}}\prn{\vec{u}}= W u^\mu + \cts{\P}{^{\mu}}\prn{\vec{u}}
  $$
  where $\cts{\P}{^\alpha}\prn{\vec{u}}$ is the part of $\ct{\P}{^\alpha}$ orthogonal to $\ct{u}{^{\alpha}}$, called the \emph{asymptotic super-Poynting} vector field.
  Observe that $\ct{n}{_{\mu}}\ct{\P}{^{\mu}}\prn{\vec{u}}=0$ is equivalent to $\ct{n}{_{\mu}}\cts{\P}{^{\mu}}\prn{\vec{u}}=0$, so that condition \eqref{eq:condition} can be equivalently formulated with just $\ct{n}{_{\mu}}\cts{\P}{^{\mu}}\prn{\vec{u}}=0$. (For $\Lambda >0$, $\cts{\P}{^{\mu}}\prn{\vec{u}}=0$ is the criterion for absence of radiation with $u^\alpha$ the unit timelike normal to $\scri$, \cite{Fernandez-Alvarez-Senovilla2020b,Fernandez-Alvarez-Senovilla2022b, Senovilla2022}).
  \end{remark}
  \begin{remark}
	  Notice that the criterion is formulated pointwise, but it can be applied to extended regions of $\scri$, such as cuts ---two-dimensional closed spacelike surfaces--- or entire connected regions $\Delta\subset \scri$.
  \end{remark}

  \Cref{crit:grav-rad} determines the absence (presence) of gravitational radiation at $\scri$, as it detects the flux of `tidal energy' \cite{Bel1958,Fernandez-Alvarez-Senovilla2022a} in directions \emph{transversal} to the conformal boundary as measured by observers within that boundary. This flux is caused by the arrival (departure) of gravitational radiation at (from) $\scri$. Indeed, it is possible to further characterise gravitational radiation by  distinguishing incoming (entering  $\hat{M}$ from $\scri$) and outgoing (arriving at $\scri$ from $\hat{M}$) gravitational radiation. For that, a second criterion is presented,
	\begin{crit}[No outgoing (ingoing) gravitational radiation with $\Lambda<0$]\label{crit:no-outgoing-gw}
 	 There is no \emph{outgoing} (\emph{ingoing}) gravitational radiation at $\scri$ if and only if
		\begin{equation}\label{eq:outgoing-condition}
			\ct{n}{_{\mu}}\ct{\P}{^{\mu}}\prn{\vec{u}}\geq 0\ \ \ (\ct{n}{_{\mu}}\ct{\P}{^{\mu}}\prn{\vec{u}}\leq 0) \ \ \ \ \ \forall\ \ct{u}{^{\alpha}}, \ \ \  \ct{u}{^{\mu}}\ct{u}{_{\mu}}=-1\ ,\ct{u}{^{\mu}}\ct{n}{_{\mu}}=0\ .
		\end{equation}	 
	\end{crit}
		\begin{remark}
 		 This tells that there is no flux of superenergy escaping from (entering to) the space-time. Thus, any possible flux is due to incoming (outgoing) gravitational radiation. 
		\end{remark}
		
		\begin{remark}
		Observe that both criteria are independent of the matter contents of the physical spacetime as long as \eqref{T->0} holds.
		\end{remark}
 
\section{Implications and the covariant formulation}\label{sec:implications}
	Criterion \ref{crit:grav-rad} considers \emph{a family} of supermomenta, and this is necessary due to the many different observers that can ``measure'' flux of tidal energy at a timelike $\scri$. As we will see later, this also fits perfectly with the set-up for an initial-boundary value problem that provides $\prn{M,\ct{g}{_{\alpha\beta}}}$. Furthermore, as we are going to prove next, the criterion can be restated in a fully geometric manner using only the rescaled Weyl tensor and the unit normal to $\scri$. To that end, first define at $\scri$ the tensor fields
	 	\begin{align}
	 	    \ct{C}{_{\alpha\beta}} &\defeqs \ct{n}{^{\mu}}\ct{n}{^{\nu}}\ctr{^*}{d}{_{\alpha\mu\beta\nu}} \ ,\label{eq:n-magnetic}\\
	 	    \ct{D}{_{\alpha\beta}} &\defeqs \ct{n}{^{\mu}}\ct{n}{^{\nu}}\ct{d}{_{\alpha\mu\beta\nu}}\ .\label{eq:n-electric}
	 	\end{align}
	The two tensor fields \eqref{eq:n-magnetic} and \eqref{eq:n-electric} are symmetric, traceless, orthogonal to $n^\alpha$, and are univocally and geometrically defined. They determine the rescaled Weyl tensor at $\scri$ by means of (see Appendix \ref{App:decom})
	$$
	d_{\alpha\beta\lambda\mu} =(g_{\alpha\beta\rho\sigma} g_{\lambda\mu\tau\nu} -\eta_{\alpha\beta\rho\sigma} \eta_{\lambda\mu\tau\nu})n^\rho n^\tau D^{\sigma\nu} - (g_{\alpha\beta\rho\sigma} \eta_{\lambda\mu\tau\nu} -\eta_{\alpha\beta\rho\sigma} g_{\lambda\mu\tau\nu})n^\rho n^\tau C^{\sigma\nu} ,
	$$
	where $g_{\alpha\beta\lambda\mu}= g_{\alpha\lambda} g_{\beta\mu} -g_{\alpha\mu} g_{\beta\lambda}$. Then, one has the following theorem
	 	\begin{thm}[No gravitational radiation condition with $\Lambda<0$]\label{thm:linearly-dependent}
	 		Let $\prn{M,\ct{g}{_{\alpha\beta}}}$ be the conformal completion of a physical space-time with negative cosmological constant $\Lambda<0$, with conformal boundary $\scri$ whose unit normal is $\ct{n}{_{\alpha}}$. Also, define the decomposition of the rescaled Weyl tensor $\ct{d}{_{\alpha\beta\gamma}^{\delta}}$ at $\scri$ with respect to $\ct{n}{_{\alpha}}$ as in \cref{eq:n-magnetic,eq:n-electric}. Then, there is no gravitational radiation at $\scri $ if and only if  $\ct{C}{_{\alpha\beta}}$ and $\ct{D}{_{\alpha\beta}}$ are pointwise linearly dependent, i.e., there exist real functions $\beta$ and $\gamma$ on $\scri$ such that
	 			\begin{equation}\label{eq:condition-prop-covariant}
	 				\beta \ct{D}{_{\alpha\beta}}=\gamma \ct{C}{_{\alpha\beta}}\ , \hspace{5mm} \gamma^2+\beta^2 \neq 0\ .
	 			\end{equation}
	 	\end{thm}
	 	\begin{proof}	
	 	 From the formula \eqref{split} in Appendix \ref{App:decom} we know that the orthogonal splitting of $ \ct{\D}{_{\alpha\beta\gamma\delta}}$ with respect to $\ct{n}{_{\alpha}}$ reads
		 \begin{equation}\label{Dsplit}
		 \ct{\D}{_{\alpha\beta\gamma\delta}}= \W \,  n_\alpha n_\beta n_\gamma n_\delta +4 \Q_{(\alpha} n_\beta n_\gamma n_{\delta)}+6t_{(\alpha\beta}n_\gamma n_{\delta)}+ 4 \Q_{(\alpha\beta\gamma} n_{\delta)} + t_{\alpha\beta\gamma\delta}\ ,
		 \end{equation}
where, in particular (formulas \eqref{q}, \eqref{Q} and \eqref{projector})
\begin{eqnarray}
\Q_{\mu} \defeq \ct{\D}{_{\alpha\beta\gamma\rho}} n^\alpha n^\beta n^\gamma h^\rho_\mu  = 2\epsilon_{\mu\rho\sigma} C^{\tau \rho} D^\sigma{}_\tau, \hspace{3mm} \Q_{\gamma}=-\Q^\rho{}_{\rho\gamma}, \hspace{3mm} n^\alpha \Q_{\alpha} =0, \label{Qa}\\
\Q_{\lambda\mu\nu} \defeq \ct{\D}{_{\alpha\beta\gamma\delta}} n^\alpha h^\beta{}_\lambda h^\gamma{}_\mu h^\delta{}_\nu = -\Q_\lambda h_{\mu\nu} +2\epsilon_{\lambda}{}^{\rho\sigma} (D_{\mu\sigma} C_{\nu\rho} +D_{\nu\sigma} C_{\mu\rho} ),\label{Qabc} \\
Q_{\lambda\mu\nu}=\Q_{(\lambda\mu\nu)}, \hspace{2mm} n^\nu \Q_{\lambda\mu\nu}=0.
\end{eqnarray}
Let $u^\alpha$ be any unit timelike vector tangent to $\scri$, that is, $u^\alpha n_{\alpha} =0$. Then, from \eqref{Dsplit} one readily gets
$$
\ct{\P}{^{\mu}}\prn{\vec{u}} =\Q_{\alpha\beta\gamma} u^\alpha u^\beta u^\gamma n^\mu + u^\alpha u^\beta u^\gamma t_{\alpha\beta\gamma}{}^\mu 
$$
and, given that $t_{\alpha\beta\gamma\delta}$ is fully orthogonal to $n^\alpha$, the criterion condition $	\ct{n}{_{\mu}}\ct{\P}{^{\mu}}\prn{\vec{u}} =0$ becomes simply
	 		\begin{equation}\label{eq:cond-Q}
	 			\ct{\Q}{_{\alpha\beta\gamma}}\ct{u}{^{\alpha}}\ct{u}{^{\beta}}\ct{u}{^{\gamma}}=0 
	 		\end{equation}
for all unit timelike $u^\alpha$ orthogonal to $n^\alpha$. But $\ct{Q}{_{\alpha\beta\gamma}}$ is fully symmetric and orthogonal to $n^\alpha$ too, so that 
		\cref{eq:cond-Q} is equivalent to 
		$$\ct{\Q}{_{\alpha\beta\gamma}}=0.$$ 
From the second in \eqref{Qa} this entails also $\Q_\alpha=0$. 
Taking this into account \eqref{Qabc} then leads to
$$
\epsilon_{\lambda}{}^{\rho\sigma} (D_{\mu\sigma} C_{\nu\rho} +D_{\nu\sigma} C_{\mu\rho} )=0\ ,
$$
or equivalently 
$$
D_{\mu[\sigma} C_{\rho]\nu} +D_{\nu[\sigma} C_{\rho]\mu} =0 \ ,
$$
which can only happen if and only if \eqref{eq:condition-prop-covariant} holds, as can be seen by projecting to an orthonormal basis and considering all possible different cases.
	 	\end{proof} 
\begin{remark}
The condition \eqref{eq:condition-prop-covariant} includes in particular the case with $C_{\alpha\beta}=0$ (for $\beta=0$) and the case with $D_{\alpha\beta}=0$ (for $\gamma =0$). The trivial case with $C_{\alpha\beta}=D_{\alpha\beta}=0$ arises when the rescaled Weyl tensor vanishes and obviously no radiation is present.
\end{remark}
\begin{remark}\label{sabrina}
In a recent paper \cite{Ciambelli2024}, $\Q_\alpha=0$ was proposed as a non-radiation condition at $\scri$ with negative $\Lambda$. This was based on the analogy with the criterion for no radiation that we presented in \cite{Fernandez-Alvarez-Senovilla2022b,Fernandez-Alvarez-Senovilla2020b} for the case with positive $\Lambda$, but without taking into account the reasons behind that criterion: vanishing of the physical flux of tidal energy as seen by the privileged observer defined by $\scri$. As we have just seen the condition $\Q_\alpha =0$ is only necessary in the case of a negative $\Lambda$ and one needs the stronger requirement of $\Q_{\alpha\beta\gamma}=0$, or equivalently the  proportionality shown in \eqref{eq:condition-prop-covariant}. This is based on the same physical criteria: vanishing of the transverse tidal energy flux as seen by observers geometrically defined by $\scri$.
\end{remark}
\begin{remark}\label{remarkgradient}
 Given that $C_{\alpha\beta}$ is divergence free, condition \eqref{eq:condition-prop-covariant} implies that in the absence of radiation (here ${\cal D}$ is the covariant derivative within $\scri$)
$$
{\cal D}^\alpha\left(\beta/\gamma\right) D_{\alpha\beta} + \beta/\gamma \, {\cal D}^\alpha D_{\alpha\beta} =0 .
$$
This restricts the asymptotic behaviour of the physical energy-momentum tensor \eqref{T->0}, which is related to ${\cal D}^\alpha D_{\alpha\beta} $. If $\beta/\gamma$ is constant, $D_{\alpha\beta}$ is divergence-free, compatible with vacuum. If on the other hand the physical spacetime is vacuum near $\scri$, the previous relation only requires that the gradient of $\beta/\gamma$ be an eigenvector field of $D_{\alpha\beta}$ ---and thus also of $C_{\alpha\beta}$--- with vanishing eigenvalue.

\end{remark}
\begin{remark}
Conditions considered in \cite{Mukhopadhyay2013} are particular cases of \cref{eq:condition-prop-covariant} with $\gamma/\beta$ constant, ergo having a divergence-free $\ct{D}{_{\alpha\beta}}$, and were used to reconstruct the bulk geometry as the stationary external region of black-hole vacuum solutions. This as well as other studies on AdS/CFT building the physical spacetime such as  \cite{GathEtal2015} and \cite{Hubeny2009} seem to agree with our criterion.
\end{remark}

Notice that, by using the same calculation leading to \cref{eq:cond-Q}, Criteria \ref{crit:no-outgoing-gw} require that the tensor $\Q_{\alpha\beta\gamma}$, which is tangent to $\scri$, has the property 
			\begin{equation}\label{property}
			\Q_{\alpha\beta\gamma} u^\alpha u^\beta u^\gamma \geq 0 \hspace{1cm} (\Q_{\alpha\beta\gamma} u^\alpha u^\beta u^\gamma \leq 0 )
			\end{equation}
			for all timelike vectors $u^\alpha$. This brings to mind the weak energy condition on the energy-momentum tensor, but now for a fully symmetric and traceless 3-index tensor. 
			As $\Q_{\alpha\beta\gamma} $ is fully symmetric and tangent to $\scri$, it is enough, for example, that it is proportional to the tensor product of  three copies of the same {\em null vector} field. This can be achieved by choosing the tensor fields $C_{\alpha\beta}$ and $D_{\alpha\beta}$ judiciously. An explicit example is given by
			\begin{equation}\label{example}
			D_{\alpha\beta} =X L_\alpha L_\beta, \hspace{1cm} C_{\alpha\beta} =Z \left( L_\alpha p_\beta +L_\beta p_\alpha\right)
			\end{equation} 
where $L^\alpha$ is a null future-pointing vector field within $\scri$ and $p^\alpha$ is a unit spacelike vector field tangent to $\scri$ and orthognal to $L^\alpha$. $X$ and $Z$ are functions on $\scri$. By using formulas \eqref{Qa} and \eqref{Qabc} it is easy to get then
\begin{equation}\label{Q=LLL}
\Q_{\alpha\beta\gamma} = 4 X Z L_\alpha L_\beta L_\gamma\ ,
\end{equation}
which obviously satisfies the property \eqref{property} as long as $XZ$ is a non-positive (or non-negative) function on $\scri$. An alternative choice leading to the very same expression \eqref{Q=LLL} is obviously
\begin{equation}\label{example1}
			C_{\alpha\beta} = X L_\alpha L_\beta, \hspace{1cm} D_{\alpha\beta} =Z \left( L_\alpha p_\beta +L_\beta p_\alpha\right).
			\end{equation} 
Of course, as explained before in remark \ref{remarkgradient}, these tensor fields $C_{\alpha\beta}$ and $D_{\alpha\beta}$ are subject to some differential conditions on $\scri$, namely, $C_{\alpha\beta}$ is divergence-free and the divergence of $D_{\alpha\beta}$ depends on the asymptotic behaviour of the energy-momentum tensor of the physical spacetime. In principle this will not be an obstacle to have conditions such as \eqref{example} or \eqref{example1}. For instance, in the latter case the divergence-free property of $C_{\alpha\beta}$ will require that $L^\alpha$ be a geodesic vector field. This can be achieved by simply selecting a null $L^\alpha$ at an initial 2-surface and then extend it uniquely by the geodesic condition. Concerning the form of $D_{\alpha\beta}$, it will depend on the matter content and its properties, but it will generically lead to some differential equations for $Z$ and $p^\alpha$ within $\scri$ that can be satisfied.

We would like to stress now,  in relation to  remark \ref{sabrina}, that these two examples \eqref{example} and \eqref{example1} have in particular $\Q_\mu=0$, that is, the condition advocated in \cite{Ciambelli2024}, but as we see here they actually do have radiation, either ingoing or outgoing ---or even both if the function $XZ$ changes sign.

\subsection{The criterion in terms of rescaled Weyl scalars}\label{rmk:weyl-scalars}
	  	  The no-radiation condition \eqref{eq:condition} can be expressed in terms of rescaled Weyl scalars $\phi_i$ ($i=0,1,2,3,4$), i.e., the five complex components of the rescaled Weyl tensor $\ct{d}{_{\alpha\beta\gamma}^{\delta}}$ in a null tetrad 
		  $$\{\ell^\mu,k^\mu,m^\mu,\bar{m}^\mu\}$$ 
with the typical non-zero scalar products $\ell^\mu k_\mu =-m^\mu \bar{m}_\mu =-1$ (see e.g. \cite{Stephani2003}, the bar denotes complex conjugate). The simplest expression follows by choosing a tetrad with real null directions\footnote{Here and everywhere else, we will always implicitly assume that all null and timelike vectors are future oriented.} defined as
	  	  	\begin{align}
	  	  	    \ct{\ell}{^\alpha} &\defeq \frac{1}{\sqrt{2}}\prn{\ct{u}{^{\alpha}}-\ct{n}{^{\alpha}}}\ ,\label{eq:ellcoplanar} \\
	  	  	    \ct{k}{^\alpha} & \defeq \frac{1}{\sqrt{2}}\prn{\ct{u}{^{\alpha}}+\ct{n}{^{\alpha}}}\ .\label{eq:kcoplanar}
	  	  	\end{align}
	  	  Then,  \cref{eq:condition} (no gravitational radiation) is equivalent to
	  	  	\begin{equation}\label{eq:rad-cond-tetrad-coplanar}
		  	  	\phi_0\bar{\phi}_{0}-\phi_4\bar{\phi}_{4}+2\prn{\phi_1\bar{\phi}_{1}-\phi_3\bar{\phi}_{3}}=0\ \ \ \ \ \ 
				\forall\ \ct{u}{^{\alpha}}, \ \ \  \ct{u}{^{\mu}}\ct{u}{_{\mu}}=-1\ ,\ct{u}{^{\mu}}\ct{n}{_{\mu}}=0\ .
	  	  	\end{equation}
		  	 Observe that the $\phi$'s appearing in \cref{eq:rad-cond-tetrad-coplanar} depend on the unit, future-pointing and timelike $\ct{u}{^{\alpha}}$. By changing the $u^\alpha$ many different conditions of the same type can be found and then stricter restrictions on the re-scaled Weyl tensor arise. In particular, we can prove that the above quadratic properties \eqref{eq:rad-cond-tetrad-coplanar} are equivalent to a set of linear relations when considered valid for all $u^\alpha$ tangent to $\scri$. To prove that, note that condition \eqref{eq:condition-prop-covariant} reads
\begin{equation}\label{eq:condition-propto}
\beta \left(\ct{k}{^\mu}\ct{k}{^{\nu}}\ct{d}{_{\alpha\mu\beta\nu}}-\ct{\ell}{^\mu}\ct{k}{^{\nu}}\ct{d}{_{\alpha\mu\beta\nu}}- \ct{k}{^\mu}\ct{\ell}{^{\nu}}\ct{d}{_{\alpha\mu\beta\nu}}+\ct{\ell}{^\mu}\ct{\ell}{^{\nu}}\ct{d}{_{\alpha\mu\beta\nu}}\  \right) =\gamma\epsilon_{\alpha\rho\sigma} 
\ct{d}{^{\rho\sigma}_{\beta\nu}} (k^\nu -\ell^\nu)/\sqrt{2}
\end{equation} 
where $\epsilon_{\alpha\mu\nu}$ is defined in formula \eqref{neta} of the Appendix. Contracting here with $k^\beta$ and $(k^\alpha+\ell^\alpha)$ one readily gets 
$$
\beta\,  \ct{k}{^\mu}\ct{\ell}{^{\alpha}}\ct{k}{^{\nu}}\ct{\ell}{^{\beta}}\ct{d}{_{\alpha\mu\beta\nu}} =-i\, \gamma m^\alpha \bar{m}^\mu \ell^\beta k^\nu \ct{d}{_{\alpha\mu\beta\nu}}
$$
or, using the standard notation for the rescaled Weyl-tensor scalars \cite{Stephani2003} 
\begin{equation}\label{lin0}
\beta {\rm Re} (\phi_2)=- \gamma {\rm Im}(\phi_2) .
\end{equation}
Similarly, contracting \eqref{eq:condition-propto} with $k^\beta$ and $m^\alpha$ one gets
$$
\beta \left(-m^\alpha \ct{k}{^\mu} k^\beta \ct{\ell}{^{\nu}}\ct{d}{_{\alpha\mu\beta\nu}} + m^\alpha \ct{\ell}{^\mu} k^\beta \ct{\ell}{^{\nu}}\ct{d}{_{\alpha\mu\beta\nu}}\right) =-i \gamma \left(m^\alpha \ct{k}{^\mu} k^\beta \ct{\ell}{^{\nu}}\ct{d}{_{\alpha\mu\beta\nu}} + m^\alpha \ct{\ell}{^\mu} k^\beta \ct{\ell}{^{\nu}}\ct{d}{_{\alpha\mu\beta\nu}}\right)
$$
or in the scalar notation
\begin{equation}\label{lin2}
\beta(\phi_1 +\bar\phi_3) =i \gamma (\phi_1 -\bar\phi_3) .
\end{equation}
Finally, contracting \eqref{eq:condition-propto} with $m^\beta$ and $m^\alpha$ one obtains
$$
\beta \left(-m^\alpha \ct{k}{^\mu} m^\beta \ct{k}{^{\nu}}\ct{d}{_{\alpha\mu\beta\nu}} + m^\alpha \ct{\ell}{^\mu} m^\beta \ct{\ell}{^{\nu}}\ct{d}{_{\alpha\mu\beta\nu}}\right) =i \gamma \left(m^\alpha \ct{k}{^\mu} m^\beta \ct{k}{^{\nu}}\ct{d}{_{\alpha\mu\beta\nu}} - m^\alpha \ct{\ell}{^\mu} m^\beta \ct{\ell}{^{\nu}}\ct{d}{_{\alpha\mu\beta\nu}}\right)
$$
or equivalently
\begin{equation}\label{lin1}
	\beta ( \phi_0+ \bar\phi_4) = \gamma i (\phi_0 -\bar\phi_4) . 
	\end{equation}
The contraction of \eqref{eq:condition-propto} with $\bar m^\beta$ and $m^\alpha$ leads again to \eqref{lin0}. 

Relations \eqref{lin0}, \eqref{lin1} and \eqref{lin2} lead immediately to the following set of linear conditions		 
			 \begin{eqnarray}
			 \alpha\bar\phi_4 = -\bar\alpha \phi_0 , \label{linear1}\\
			  \alpha \bar\phi_1 =-\bar\alpha \phi_3, \label{linear}\\
			 \alpha \bar\phi_2 =-\bar\alpha \phi_2 .\label{linear2}
			 \end{eqnarray}
			 where $\alpha := \beta +i\gamma\neq 0$ is non-vanishing complex function on $\scri$, with $\beta,\gamma$ given in \eqref{eq:condition-prop-covariant}. It is remarkable that $\phi_2$ does not arise explicitly in \eqref{eq:rad-cond-tetrad-coplanar} and yet it is restricted due to the validity of \eqref{eq:rad-cond-tetrad-coplanar} for all unit timelike vectors tangent to $\scri$. 

The version of conditions \eqref{example} and \eqref{example1} using rescaled Weyl scalars is discussed in section \ref{sec:IBVP} on boundary conditions at $\scri$.

\section{Restrictions on PNDs at $\scri$}
 	  	 	 To derive the restrictions that absence of radiation impose in terms of the Principal Null Directions (PND) of the rescaled Weyl tensor an explicit expression for the left-hand side of \cref{eq:condition} can be formulated. The derivation of such formula is long and the details can be found in \cref{sec:appendix}. Here the final result is presented:
 	 	\begin{align}
 	 		\ct{n}{_{\mu}}\ct{\P}{^\mu}\prn{\vec{u}}=4\bar{\ubar{\phi}_3}\ubar{\phi}_3&\Bigg\lbrace B_1\brkt{\ubar{a}\breve{f}\bar{\breve{f}}\prn{\frac{1}{3}\hat{\ubar{a}}\hat{\ubar{b}}+\hat{\ubar{d}}\bar{\hat{\ubar{d}}}}+\ubar{b}\prn{\frac{1}{3}\breve{\ubar{a}}\breve{\ubar{b}}+\breve{\ubar{d}}\bar{\breve{\ubar{d}}}}+\ubar{\breve{b}}\prn{\frac{1}{3}{\ubar{a}}{\ubar{b}}+{\ubar{d}}\bar{{\ubar{d}}}}}\nonumber\\
 	 		&+B_2\breve{c}\bar{\breve{c}}\brkt{\ubar{a}\tilde{f}\bar{\tilde{f}}\prn{\frac{1}{3}{\ubar{a}^\prime}{\ubar{b}^\prime}+{\ubar{d^\prime}}\bar{{\ubar{d}^\prime}}}+\ubar{b}\prn{\frac{1}{3}\tilde{\ubar{a}}\tilde{\ubar{b}}+\tilde{\ubar{d}}\bar{\tilde{\ubar{d}}}}+\ubar{\tilde{b}}\prn{\frac{1}{3}{\ubar{a}}{\ubar{b}}+{\ubar{d}}\bar{{\ubar{d}}}}}\Bigg\rbrace\nonumber\\
 	 		+4\phi_{3}\bar{\phi_{3}}&\Bigg\lbrace\tilde{B}_{3}\brkt{{a}\breve{c}\bar{\breve{c}}\prn{\frac{1}{3}\hat{{a}}\hat{{b}}+\hat{{d}}\bar{\hat{{d}}}}+{b}\prn{\frac{1}{3}\breve{{a}}\breve{{b}}+\breve{{d}}\bar{\breve{{d}}}}+{\breve{b}}\prn{\frac{1}{3}{{a}}{{b}}+{{d}}\bar{{{d}}}}}\nonumber\\
 	 		&+\breve{B}_4\brkt{{a}\tilde{c}\bar{\tilde{c}}\prn{\frac{1}{3}{{a}^\prime}{{b}^\prime}+{{d^\prime}}\bar{{{d}^\prime}}}+{b}\prn{\frac{1}{3}\tilde{{a}}\tilde{{b}}+\tilde{{d}}\bar{\tilde{{d}}}}+{\tilde{b}}\prn{\frac{1}{3}{{a}}{{b}}+{{d}}\bar{{{d}}}}}\Bigg\rbrace\ .\label{eq:formula270}
 	 	\end{align}
 	
		 \begin{figure}[h!]
 	 	 	 \centering
 	 	 	 \includegraphics[scale=1]{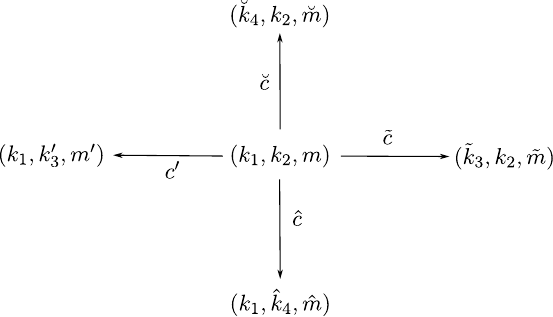}
 	 	 	 \caption{{\small Schematic representation of the first set of tetrads and their transformations used to derive \cref{eq:formula270}. The null directions, labelled with $i=1,2,3,4$, are aligned with the four PNDs of the rescaled Weyl tensor $(k_1,k_2,k_3,k_4)$. The starting basis is $\prn{{k}{_{1}^\alpha},{k}{_{2}^\alpha}, \ct{m}{^\alpha}}$, and the rest of the tetrads are obtained by doing a null rotation along $k_1^\alpha$ (with complex coefficients $\hat{c}$ and $c^\prime$) or $k_2^\alpha$ (with complex coefficientes $\tilde{c}$ and $\breve{c}$). The different decorations indicate relative normalisations (e.g., $k^\prime_3$ and $\tilde{k}_{3}$ are proportional and aligned with the PND $k_3$). See \cref{sec:appendix} for full expressions.}}\label{fig:set1}
 	  	 \end{figure}
		 
 	  	 \begin{figure}[h!]
 	 	 	 \centering
 	 	 	 \includegraphics[scale=1]{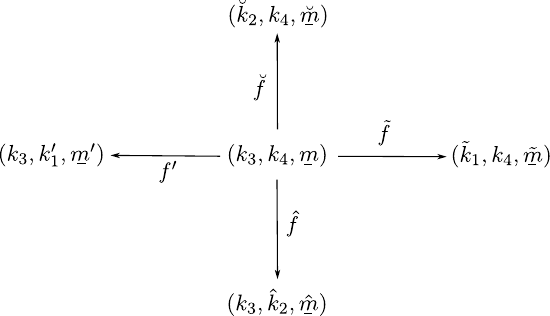}
 	 	 	 \caption{{\small Similarly to \cref{fig:set1}, this represents a second set of null tetrads. In this case, the starting point is $\prn{{k}{_{3}^\alpha},{k}{_{4}^\alpha}, \underset{\bar{}}{m}^\alpha}$ and $f^\prime$, $\hat{f}$, $\tilde{f}$, $\breve{f}$ are the complex coefficients of the different null rotations. The bar appearing below some of the complex vectors in the tetrad is used just to distinguish them from those of the first set in \cref{fig:set1} (that is, to distinguish $\underset{\bar{}}{m}^\alpha$   from $\ct{m}{^\alpha}$).}}\label{fig:set2}
 	  	 \end{figure}
		 
		  All the scalars appearing in this equation are related to one of the null tetrads described in \cref{fig:set1,fig:set2}. In particular, $\ubar{\phi}_3$ and $\phi_3$ are the corresponding rescaled Weyl scalars in the basis defined by  $\prn{k^\alpha_3,k^\alpha_4}$ and $\prn{k^\alpha_{1},k^\alpha_{2}}$, respectively, and $k^\alpha_{i}$ ($i=1,2,3,4$) are the PND of the rescaled Weyl tensor (at $\scri$).  The real $a$, $b$ and complex $d$ coefficients and their decorated versions are the components of the generic unit timelike vector field $\ct{u}{^\alpha}$ ---with $u^\alpha n_\alpha =0$--- in each of the bases appearing in \cref{fig:set1,fig:set2}. In particular, $a$, $b$ and their decorated versions are always strictly positive. For instance, one has:
		  \begin{align}
		     \ct{u}{^\alpha}&=a{k}{_1^{\alpha}}+b{k}{_2^\alpha}+d\ct{m}{^\alpha}+\bar{d}\ct{\bar{m}}{^\alpha}\nonumber\\	
		       & = \tilde{a}\tilde{k}_{3}^\alpha+\tilde{b}k_{2}^\alpha+\tilde{d}\tilde{m}^\alpha+\bar{\tilde{d}}\bar{\tilde{m}}^\alpha	\nonumber\\
		       & = \ubar{a}{k}_{3}^\alpha+\ubar{b}k_{4}^\alpha+\ubar{d}\ubar{m}^\alpha+\bar{\ubar{d}}\bar{\ubar{m}}^\alpha\ .
		  \end{align}
		  
		   Detailed relations and definitions are found in \cref{sec:appendix}, in particular there are many simplifications if the different Petrov types of the re-scaled Weyl tensor are special, as then two or more of the PNDs coincide. These simplifications in algebraically special cases are considered at large in the appendix.  		

 The most important feature of \cref{eq:formula270} is that all terms are manifestly positive, except for $B_1$, $B_2$, $\tilde{B}_{3}$ and $\breve{B}_{4}$, defined by 
 	 	\begin{equation}
 	 	B_1\defeq {k}{^\mu_{1}}\ct{n}{_{\mu}}\ ,\quad B_2\defeq {k}{^\mu_{2}}\ct{n}{_{\mu}}\ ,\quad \tilde{B}_3\defeq {\tilde{k}}{^\mu_{3}}\ct{n}{_{\mu}}\ ,\quad \breve{B}_4\defeq \breve{k}{^\mu_{4}}\ct{n}{_{\mu}}\ ,
 	 	\end{equation}
		which can have any sign depending on the \emph{relative orientation} between the unit normal $\ct{n}{_{\alpha}}$ and the PNDs.
		
		First of all, we prove an important result showing that, {\em in absence of radiation} , the unit normal and the PNDs aligned themselves in very specific geometrical ways.
\begin{lemma}\label{lemma:pnd-pair-scri}
Let $\prn{M,\ct{g}{_{\alpha\beta}}}$ be the conformal completion of a physical space-time with negative cosmological constant $\Lambda<0$, with conformal boundary $\scri$ whose unit normal is $\ct{n}{_{\alpha}}$. Then, if condition \eqref{eq:condition} holds,  principal null directions that are not tangent to $\scri$ always come in pairs, and $\ct{n}{_{\alpha}}$ is contained in the plane they define.
\end{lemma}
\begin{proof}
Recall that \eqref{eq:condition} is equivalent to \eqref{eq:condition-prop-covariant}, which in turn can be re-written in a null tetrad  coplanar with $\ct{n}{^{\alpha}}$ ---of type \cref{eq:ellcoplanar,eq:kcoplanar}--- as in \eqref{lin0}--\eqref{linear2}. Next, \emph{assume there is one PND, say $k_1^\alpha$, not tangent to $\scri$} --so  it has $B_1\neq 0$-- and align the tetrad vector $\ct{\ell}{^\alpha}$ with it,
		\begin{equation*}
			B_1\sqrt{2}\ct{\ell}{^{\alpha}}=\ct{k}{_{1}^{\alpha}}\ .
		\end{equation*} 
		As $\ell^\alpha$ is a PND, one has $\phi_0=0$. 
		Then,  \eqref{lin1} reads simply $\beta\bar{\phi}_4 =-i\gamma \bar{\phi}_4$ which, as $\beta$ and $\gamma$ are real, leads directly to $\phi_4 =0$. This implies that $\ct{k}{^{\alpha}}$ is a PND too.
\end{proof}

 	\Cref{eq:formula270} and Lemma \ref{lemma:pnd-pair-scri} are  used to prove another result on the absence of gravitational radiation in terms of PNDs and the Petrov classification of the rescaled Weyl tensor at $\scri$:
	 \begin{thm}[No gravitational radiation condition with $\Lambda<0$ and algebraic classification]\label{thm:nograv-algebraic}
		 Let $\prn{M,\ct{g}{_{\alpha\beta}}}$ be the conformal completion of a physical space-time with negative cosmological constant $\Lambda<0$, with conformal boundary $\scri$ whose unit normal is $\ct{n}{_{\alpha}}$. Then, there is no gravitational radiation at the conformal boundary $\scri$ if and only if one of the next cases hold:
		 	\begin{enumerate}
		 	\item \label{it:case1} All principal null directions $\ct{k}{_{i}^\alpha}$ of the rescaled Weyl tensor $\ct{d}{_{\alpha\beta\gamma}^{\delta}}$ are tangent to $\scri$.
		 	\item \label{it:case2} $\ct{n}{_{\alpha}}$ is coplanar with two principal null directions of the same multiplicity and the remaining PNDs (if any) are tangent to $\scri$.	
			\item \label{it:case3} At $\scri$, $\ct{d}{_{\alpha\beta\gamma}^{\delta}}$ is of Petrov type I and $\ct{n}{_{\alpha}}$ is coplanar with a pair of principal null directions, and simultaneously coplanar with the remaining pair of PNDs. In this case, $\ct{d}{_{\alpha\beta\gamma}^{\delta}}$ is necessarily purely electric or purely magnetic with respect to the unique principal timelike direction, which is tangent to $\scri$.		
		 	\item \label{it:case4}The rescaled Weyl tensor vanishes.
		 	\end{enumerate}
	 \end{thm}

\begin{remark}
Observe that case \ref{it:case1} is feasible for all possible Petrov types of $\ct{d}{_{\alpha\beta\gamma}^{\delta}}$, and that this is the only possible case if the Petrov type is N or III. For type D and II there is another possibility, in the former case if $n^\mu$ is coplanar with the multiple PNDs, and in the latter if $n^\mu$ is coplanar with the two single PNDs and the multiple one is tangent to $\scri$ (case \ref{it:case2}). Finally, if $\ct{d}{_{\alpha\beta\gamma}^{\delta}}$ has Petrov type I, all cases \ref{it:case1}, \ref{it:case2} and \ref{it:case3} may arise.
\end{remark}
	 \begin{proof}
		 First, if case \ref{it:case1} holds, $B_1$, $B_2$, $\tilde{B}_{3}$ and $\breve{B}_{4}$ in \cref{eq:formula270} vanish, and therefore  $\ct{n}{_{\mu}}\ct{\P}{^{\mu}}\prn{\vec{u}}=0$ independently of the values of the positive $a'$s and $b'$s and of the $d'$s, that is to say, for $\ \forall\ \ct{u}{^{\alpha}}$ tangent to $\scri$, and thus there is no gravitational radiation. 
		To prove that case \ref{it:case2} also leads to absence of radiation, we assume coplanarity of $k^\alpha_{1}$ and $k^\alpha_{2}$ with $\ct{n}{_{\alpha}}$ then
		 $$
		 n^\alpha = -B_2 k_1^\alpha - B_1 k_2^\alpha, \hspace{4mm} 2B_1 B_2 =-1\ ,
		 $$
		 and given that, by definition (see \eqref{eq:u-components})
		 $$\ct{u}{^\alpha}=ak^\alpha_{1}+bk^\alpha_{2}+d\ct{m}{^{\alpha}}+\bar{d}\ct{\bar{m}}{^{\alpha}}$$
 one can immediately check that the condition of $u^\alpha$ being tangent to $\scri$ reads
 $$
 \ct{n}{_{\alpha}}\ct{u}{^\alpha}=a B_1 + b B_2 =0 .
 $$ 
If the rescaled Weyl tensor has Petrov type D one can use  \cref{eq:typeD} in \cref{sec:appendix}, which is a reduction of \cref{eq:formula270} to the case of a Petrov type D $\ct{d}{_{\alpha\beta\gamma}^{\delta}}$, namely
		 	\begin{equation} \label{typeD}
		 		\ct{n}{_{\mu}}\ct{\P}{^\mu}\prn{\vec{u}}=36\phi_{2}\bar{\phi}_{2}\prn{aB_1+bB_2}\prn{ab+d\bar{d}}\ .
		 	\end{equation} 
This together with the previous gives  $\ct{n}{_{\mu}}\ct{\P}{^{\mu}}\prn{\vec{u}}=0$ for all timelike $ \ct{u}{^{\alpha}}$ tangent to $\scri$.

If the rescaled Weyl tensor has Petrov type II one can instead use  \cref{eq:typeII} in \cref{sec:appendix}, which again is the appropriate reduction of \cref{eq:formula270} to type II. Furthermore, we must assume that the double PND is tangent to $\scri$, so that $B_1=0$ and \eqref{eq:typeII} reduces to
\begin{equation}\label{paso0}
\ct{n}{_{\mu}}\ct{\P}{^\mu}\prn{\vec{u}}=18{\phi}_2\bar{{\phi}_2}\Bigg\lbrace 
			 		B_2b\prn{\hat{a}\hat{b}+\hat{d}\bar{\hat{d}}}+\breve{B}_4\hat{c}\bar{\hat{c}}{b\prn{ab+d\bar{d}}}\Bigg\rbrace\ .
\end{equation}
Further, we need to assume that $n^\alpha$ lies in the plane of the remaining two PNDs, so that
$$
n^\alpha = -B_2 \breve{k}_4^\alpha -\breve{B}_4 k_2^\alpha , \hspace{4mm} 2 B_2 \breve{B}_4 =-1.
$$
This implies in particular that $n_\alpha \breve{m}^\alpha =0 = D+\bar{\breve{c}} B_2$ which together with \eqref{eq:B4-B1B2} and with the checked version of \eqref{eq:utilde}, gives respectively
$$
\breve{B}_4 =\breve{c} D = -B_2 \breve{c} \bar{\breve{c}}, \hspace{1cm} \breve{a} \breve{B}_4 =- \breve{b} B_2 \ ,
$$
implying in particular that
$$
0 = \breve{b} -\breve{a}\breve{c} \bar{\breve{c}}= b-\breve{c} \bar{d} -\bar{\breve{c}}d =\frac{1}{\hat{c}\bar{\hat{c}}} \left(b \hat{c}\bar{\hat{c}} -\hat{c} \bar{d} -\bar{\hat{c}} d\right) .
$$
Equation \eqref{paso0} becomes then (using \eqref{eq:c-inverse} and the hatted version of \eqref{eq:utilde})
$$
\ct{n}{_{\mu}}\ct{\P}{^\mu}\prn{\vec{u}}=18{\phi}_2\bar{{\phi}_2}B_2 b \left(\hat{a}\hat{b}+\hat{d}\bar{\hat{d}}-ab-d\bar{d}\right) =
36{\phi}_2\bar{{\phi}_2}B_2 b^2 \left(b \hat{c}\bar{\hat{c}} -\hat{c} \bar{d} -\bar{\hat{c}} d\right) = 0 
$$
as required. 

Consider then the case \ref{it:case2} for Petrov type I, setting $B_1 = B_2 =0$ in the general formula \eqref{eq:formula270} and 
\begin{equation}\label{n=k3k4}
n^\alpha = -B_3 k_4^\alpha - B_4 k_3^\alpha = -\alpha_3 \alpha_4 \left(\tilde{B}_3 \breve{k}_4^\alpha +\breve{B}_4 \tilde{k}_3^\alpha\right)
\end{equation}
where we have used \eqref{eq:k3} and \eqref{eq:k4}. This immediately provides $\underline{D}=0$ which, together with $B_1=B_2=0$, in particular entails 
\begin{eqnarray}
B_4 +\hat{f}\bar{\hat{f}} B_3 &=& \breve{B}_4 \alpha_4 +\alpha_3 \hat{f}\bar{\hat{f}} \tilde{B}_3 =0,\label{B4ffB3}\\
B_4 +f'\bar{f}' B_3 &=& 0,\\
B_3 + \breve{f} \bar{\breve{f}} B_4 &=& 0,\\
B_3 + \tilde{f} \bar{\tilde{f}} B_4 &=& 0
\end{eqnarray}
providing
$$
\hat{f}\bar{\hat{f}}=f'\bar{f}' =\frac{1}{\breve{f} \bar{\breve{f}}} =\frac{1}{\tilde{f} \bar{\tilde{f}}}.
$$
Combining this with \eqref{eq:alphas} we derive
\begin{equation}\label{cs}
\tilde{c}\bar{\tilde{c}}=\breve{c}\bar{\breve{c}}
\end{equation}
and also using \eqref{eq:alphas} in the second of \eqref{B4ffB3} one obtains
\begin{equation}\label{B4B3}
\breve{B}_4 +\tilde{B}_3 =0
\end{equation}
so that \eqref{eq:rearraengement} becomes
$$
-4 \ct{n}{_{\mu}} \cts{\P}{^\mu}\prn{\vec{u}}=\phi_{3}\bar{\phi}_{3}\tilde{B}_3 \prn{J_3-J_4}
$$
hence our task is to prove that $J_3 -J_4$ vanishes. From \eqref{eq:J3} and \eqref{eq:J4} and using \eqref{cs} 
$$
J_3-J_4 = (3ab+d\bar{d})  \left[ (\tilde{c} -\breve{c})\bar{d}+(\bar{\tilde{c}} -\bar{\breve{c}}) d\right] .
$$
From \eqref{eq:tildeb} and its checked version we immediately have on using \eqref{cs}
$$
(\tilde{c} -\breve{c})\bar{d}+(\bar{\tilde{c}} -\bar{\breve{c}})= \breve{b} -\tilde{b} =-u_\alpha (\breve{k}_4^\alpha -\tilde{k}_3^\alpha) =-\frac{1}{\alpha_3 \alpha_4 \tilde{B}_3} u_\alpha n^\alpha
$$
where in the last equality use of \eqref{B4B3} and \eqref{n=k3k4} has been made. In summary, we have got
$$
\ct{n}{_{\mu}} \cts{\P}{^\mu}\prn{\vec{u}}=\frac{1}{4\alpha_3 \alpha_4}\phi_{3}\bar{\phi}_{3}\, \,  u_\alpha n^\alpha =0
$$
for all $u^\alpha$ tangent to $\scri$.

\vspace{1cm}
It remains to prove case \ref{it:case3}. By using similar arguments, one immediately gets $\ubar{D}=D=0$. This, together with \cref{eq:alphas} and \cref{eq:B3-B1B2,eq:B4-B1B2} ---and the corresponding `underbar' versions--- gives
	\begin{align}
	   \alpha_{2}B_{1}+\breve{c}\bar{\breve{c}}\alpha_{2}B_{2}-B_{4} &= 0 \ ,\\
	    \alpha_{2}\breve{f}\bar{\breve{f}}B_{1}+\alpha_{2}\breve{f}\bar{\breve{f}}\tilde{c}\bar{\tilde{c}}B_{2}-B_{3} & = 0\ ,\\
	    -B_{1}+\breve{c}\bar{\breve{c}}B_{3}\alpha_{2}+\tilde{f}\bar{\tilde{f}}\breve{c}\bar{\breve{c}}B_{4}&=0\ ,\\
	    -B_{2}+\alpha_{2}B_{3}+\breve{f}\bar{\breve{f}}\alpha_{2}B_{4}&=0\ .
	\end{align}
Then, the values of $\tilde{c}\bar{\tilde{c}}$, $\breve{c}\bar{\breve{c}}$, $\tilde{f}\bar{\tilde{f}}$ and $\breve{f}\bar{\breve{f}}$ in terms of $B_{i}$ read
	\begin{equation}\label{eq:general-sol-cs-typeI-case3}
	\tilde{c}\bar{\tilde{c}}=\frac{\alpha_{2}B_{1}B_{3}}{B_{2}^2-\alpha_{2}B_{2}B_{3}}\ ,\quad\breve{c}\bar{\breve{c}}=\frac{B_{4}-\alpha_{2}B_{1}}{\alpha_{2}B_{2}}\ ,\quad \tilde{f}\bar{\tilde{f}}=\frac{\alpha_{2}B_{1}B_{3}}{B_{4}^2-\alpha_{2}B_{1}B_{4}}\ ,\quad \breve{f}\bar{\breve{f}}=\frac{B_{2}-\alpha_{2}B_{3}}{\alpha_{2}B_{4}}\ ,	
	\end{equation}
where, in addition, one has taken into account that $-2B_{1}B_{2}=1=-2B_{3}B_{4}$. Also, orthogonality of $\ct{u}{^\alpha}$ with $\ct{n}{_{\alpha}}$ implies
	\begin{equation}\label{eq:conditions-Bs}
		B_{3}=-\frac{\ubar{b}}{\ubar{a}}B_{4}\ ,\quad B_{1}=-\frac{b}{a}B_{2}\ .
	\end{equation}
For the next steps, one takes advantage of the boost invariance of \cref{eq:general-formula} to choose\footnote{It is not necessary to do this choice, but it simplifies considerably the calculations. In previous cases it has not been used because the equations were tractable without any boost-gauge fixing.} 
	\begin{equation}
		B_1=B_3=-B_{2}=-B_{4}=\frac{1}{\sqrt{2}}\ .
	\end{equation}
Setting these values into \cref{eq:general-sol-cs-typeI-case3,eq:conditions-Bs}, respectively, shows that 
\begin{equation}\label{eq:simplifications-1}
\tilde{c}\bar{\tilde{c}}=\tilde{f}\bar{\tilde{f}}\ ,\quad\breve{c}\bar{\breve{c}}=\breve{f}\bar{\breve{f}}\ ,\quad a=b\ ,\quad \ubar{a}=\ubar{b}\ ,
\end{equation}
and this gives the relations
	\begin{equation}\label{eq:simplifications-2}
		\hat{a}=\tilde{b}\ ,\quad a^\prime=\breve{b}\ ,
	\end{equation}
which also imply (see \cref{eq:ab-dd})
	\begin{equation}\label{eq:simplifications-3}
		d^\prime \bar{d}^\prime=\breve{d}\bar{\breve{d}}\ ,\quad \hat{d}\bar{\hat{d}}=\tilde{d}\bar{\tilde{d}}\ .
	\end{equation}
With \cref{eq:simplifications-1}, it is possible to express \cref{eq:rearraengement} as ---see \cref{eq:alphas,eq:k1,eq:k2,eq:k3,eq:k4}---
	\begin{equation}
		-4\ct{n}{_{\mu}}\cts{\P}{^{\mu}}\prn{\vec{u}}=\frac{\phi_{3}\bar{\phi_{3}}}{\sqrt{2}}\brkt{J_{1}-J_{2}+\frac{1}{\alpha_{3}}\prn{J_{3}-\breve{c}\bar{\breve{c}}J_{4}}}\ ,
	\end{equation}
but if one uses now \cref{eq:simplifications-1,eq:simplifications-2,eq:simplifications-3} in \cref{eq:J3,eq:J4} and \cref{eq:J1-step,eq:J2-step}, one finds
	\begin{equation}
	J_{1}-J_{2}=0\ ,\quad J_{3}-\breve{c}\bar{\breve{c}}J_{4}=0\ .
	\end{equation}
Thus,
	\begin{equation}
	\ct{n}{_{\mu}}\cts{\P}{^{\mu}}\prn{\vec{u}}=0.
	\end{equation}
Hence, it has been proved that cases \ref{it:case1}, \ref{it:case2} and \ref{it:case3} imply absence of gravitational radiation.

\vspace{1cm}
Next, we must prove the converse. For that, recall that from \cref{thm:linearly-dependent}  $\ct{n}{_{\mu}}\ct{\P}{^{\mu}}\prn{\vec{u}}=0\ \forall\ \ct{u}{^{\alpha}}$ tangent to $\scri$ is equivalent to the covariant condition \eqref{eq:condition-prop-covariant}, and that such a condition translates, in a null tetrad with \cref{eq:ellcoplanar,eq:kcoplanar}, into the formulas derived in subsection \ref{rmk:weyl-scalars}. Notice that $n^\alpha$ is coplanar with $k^\alpha$ and $\ell^\alpha$. If there exists one PND not tangent to $\scri$, then from \cref{lemma:pnd-pair-scri} we know that we can choose the null tetrad such that both $\ell^\alpha$ and $k^\alpha$ are PNDs (as they come in pairs if $\ct{n}{_{\mu}}\ct{\P}{^{\mu}}\prn{\vec{u}}=0\ \forall\ \ct{u}{^{\alpha}}$ tangent to $\scri$).
%
%
Hence, $\phi_0=\phi_4=0$. This implies that the invariants of the rescaled Weyl tensor (see \cite{McIntoshetal1994,Stephani2003} for definitions) are 
$$
{\rm I} =3\phi_2^2 -4\phi_1 \phi_3 , \hspace{5mm} {\rm J}= \phi_2 (2\phi_1 \phi_3 -\phi_2^2)
$$	
so that the vanishing of the following invariant 
\begin{equation}\label{inv}
{\rm I}^3 -27 {\rm J}^2 = 4\phi_1^2 \phi_3^2 (9\phi_2^2 -16 \phi_1 \phi_3)
\end{equation}
determines when the rescaled Weyl tensor is algebraically special, that is to say, when there are multiple PNDs. This can happen if
\begin{itemize}
\item $\phi_1=0$, which necessarily requires $\phi_3=0$ due to  \eqref{linear}. In this case both $k^\alpha$ and $\ell^\alpha$ are multiple PNDs, leading to the case \ref{it:case2} for a rescaled Weyl tensor of Petrov type D at $\scri$  ---unless $\phi_2=0$ too, which leads to the trivial case \ref{it:case4}.
\item $\phi_1\neq 0\neq \phi_3$, but 
\begin{equation}\label{disc}
9\phi_2^2 -16 \phi_1 \phi_3=0.
\end{equation}
By using standard techniques, i.e. by solving the quadratic equation 
$$
2\phi_3 x^2 +3 \phi_2 x +2\phi_1=0 \hspace{2mm} \Longrightarrow \hspace{3mm} x=-\frac{3\phi_2}{4\phi_3}
$$
which has the {\em unique} double root shown, we know that the null vector
$$
k'^\mu:= k^\mu +x \bar m^\mu +\bar{x} m^\mu + x\bar{x} \ell^\mu 
$$
is a double PND. Therefore, this is a Petrov type II  rescaled Weyl tensor, and the normal is coplanar with the two single PNDs. It remains to show that the double PND $k'^\mu$  is tangent to $\scri$, but this follows by taking the scalar product with the normal
$$
\sqrt{2} n_\mu k'^\mu =1 - x\bar x =1 -\frac{9\phi_2 \bar\phi_2}{16\phi_3 \bar\phi_3} =0
$$
where in the last equality we have used \eqref{linear}, \eqref{linear2} and \eqref{disc}. Therefore, the double PND in this situation is tangent to $\scri$, and the case \ref{it:case2} for a rescaled Weyl tensor of Petrov type II at $\scri$ is proven.
\end{itemize}
	
There remains the case where $\phi_1\neq 0 \neq \phi_3$ and \eqref{disc} does not hold, so this is necessarily a Petrov type I situation. There are two possibilities again, because if there is a third PND not tangent to $\scri$, then as we proved before this comes in a pair with another PND such that $n^\mu$ is coplanar with these two new PNDs; otherwise the two extra PNDs must be tangent to $\scri$. The latter possibility leads to the case \ref{it:case2} for Petrov type I. In the former situation, in addition to \cref{eq:ellcoplanar,eq:kcoplanar}	 we also have the corresponding relation that places $n^\alpha$ in the plane generated by $k_3^\alpha$ and $k_4^\alpha$, that is
		$$
		n^\mu =(k^\mu-\ell^\mu)/\sqrt{2} =-B_1 k_2^\mu -B_2 k_1^\mu = -B_3 k_4^\mu -B_4 k_3^\mu, \hspace{3mm} 
		2B_1 B_2 = 2 B_3 B_4 =-1
		$$
implying the linear dependency of the four PNDs as
$$
B_1 k_2^\mu +B_2 k_1^\mu  -B_3 k_4^\mu -B_4 k_3^\mu =0.
$$
This is known to happen \cite{McIntoshetal1994} if, and only if, the rescaled Weyl tensor at $\scri$ is purely electric, or purely magnetic, with respect to a timelike vector $\vec v$ which is the unique principal timelike vector of $\ct{d}{_{\alpha\mu\beta\nu}}$, $\vec v$  is a linear combination of the  four PNDs with all coefficients strictly positive \cite{Ferrando1997}, and actually $\vec v$ is tangent to $\scri$ and has the entire asymptotic super-Poynting vanishing $ \cts{\P}{^{\mu}}\prn{\vec{v}}=0$.
		
	In summary, we have proven that the  no-gravitational radiation condition \eqref{eq:condition} implies that either case \ref{it:case2} or case \ref{it:case3} or case \ref{it:case4} of the theorem must hold if at least one of the PNDs is not tangent to $\scri$. The only remaining alternative is that all PNDs  are tangent to $\scri$, arriving at case \ref{it:case1} of the theorem.
	 \end{proof}
\begin{remark}
Notice that case \ref{it:case1} for Petrov type I also implies that the rescaled Weyl tensor is purely electric or purely magnetic, because the four PNDs are tangent to $\scri$ and this necessarily implies that they are linearly dependent.
\end{remark} 
\begin{remark}
The Petrov type of the rescaled Weyl tensor at $\scri$ should not be confused with the Petrov type of the Weyl tensor in the physical spacetime. As explained in \cite{Fernandez-Alvarez-Senovilla2022b,Senovilla2022}, the Petrov type of $\ct{d}{_{\alpha\mu\beta\nu}}$ must be equal or more specialized than that of the physical Weyl tensor and thus, for instance, if the latter is type I, all possible Petrov types are feasible for $\ct{d}{_{\alpha\mu\beta\nu}}$ at $\scri$. See \cite{Fernandez-Alvarez-Senovilla2022b,Senovilla2022} for further details.
\end{remark}

		Absence of outgoing (or incoming) gravitational radiation can be recognised geometrically too in some situations, now based on \cref{crit:no-outgoing-gw}:
	 		\begin{thm}[No outgoing/incoming gravitational radiation with $\Lambda<0$]\label{thm:outgoing-gw}
	 		 	There is no outgoing (incoming) gravitational radiation at $\scri$ if no principal null direction of the rescaled Weyl tensor $\ct{d}{_{\alpha\beta\gamma}^{\delta}}$ at $\scri$ point outwards (inwards), i.e.,
	 		 		\begin{equation}
	 		 			\ct{k}{_i^\alpha}\ct{n}{_{\alpha}}\geq 0\ \ \ \quad  (\ct{k}{_i^\alpha}\ct{n}{_{\alpha}}\leq 0)\hspace{3mm} \forall i=1,2,3,4.
	 		 		\end{equation}
	 		\end{thm}
			\begin{proof}
				This reads from \cref{eq:formula270}. If none of the $B_1$, $B_2$, $\tilde{B}_3$, $\breve{B}_4$ is negative, then $	\ct{n}{_{\mu}}\ct{\P}{^\mu}\prn{\vec{u}}\geq 0 $ for all $\ct{u}{^\alpha}$ tangent to $\scri$.
			\end{proof}

\section{Initial data and new physical boundary conditions}\label{sec:IBVP}
The initial value formulation for Einstein Field Equations (EFE) with negative $\Lambda$ takes the form of an \emph{initial boundary value problem} (IBVP). From the viewpoint of conformal spacetime, this question was tackled by Friedrich in \cite{Friedrich1995}, where he achieved a geometric and covariant formulation of the boundary conditions. There are also other more recent works in the literature, like approaches in physical space-time \cite{Friedrich-Nagy1999}, generalisations to higher dimensions \cite{Enciso-Kamran-2019}, uniqueness results using Fefferman-Graham expansions \cite{Holzegel-Shao2023,Shao2024}, and the inclusion of matter-fields content \cite{Carranza-Kroon2019}. It is \cite{Friedrich1995}, however, the work that better fits our present purposes. \\

Friedrich showed that the Cauchy data on a 3-dimensional spacelike hypersurface $\Sigma$, that is, its first $\ct{h}{_{\alpha\beta}}$ and second  $\ct{k}{_{\alpha\beta}}$ fundamental forms, together with the conformal class of metrics $\brkt{\cts{g}{_{ab}}}$ on $\scri$ determine a unique (up to diffeomorphisms), physical solution $\ct{\hat{g}}{_{\alpha\beta}}$ of the $\Lambda$-vacuum  EFE with negative cosmological constant. Thus, information about gravitational radiation must be encoded on these data, and it can be seen that \cref{crit:grav-rad} indeed puts restrictions on $\prn{\ct{h}{_{\alpha\beta}},\ct{k}{_{\alpha\beta}}}$ and on $\brkt{\cts{g}{_{ab}}}$. To understand this, first, let $\ct{u}{^\alpha}$ be the unit normal to $\Sigma$, and assume it is orthogonal to $\ct{n}{_{\alpha}}$  on $\scri$ \footnote{This construction, with the normal to the initial hypersurface tangent to $\scri$, is precisely the choice made by Friedrich in \cite{Friedrich1995}, see also \cite{Friedrich2014}.}. \Cref{crit:grav-rad} is a condition on the asymptotic super-Poyntings relative to timelike vectors $\vec u$ tangent to $\scri$ , but as is well known \cite{Bel1958,Bel1962,Senovilla2000,Alfonso2008,Fernandez-Alvarez-Senovilla2022b} any super-Poynting can be expressed as a matrix commutator of the electric and magnetic parts relative to $\vec u$. Therefore, our no-radiation criterion can be re-stated in terms of the matrix commutator of the electric $\ct{E}{_{\alpha\beta}}\prn{\vec{u}}$ and magnetic $\ct{B}{_{\alpha\beta}}\prn{\vec{u}}$ parts of the rescaled Weyl tensor ---computed with respect to $\ct{u}{^{\alpha}}$ at $\scri$--- as
	\begin{equation}\label{eq:commutator}
	u^\mu n^\nu \eta_{\mu\nu\rho\sigma} B(\vec u)^{\rho\tau} E(\vec u)^{\sigma}{}_\tau =0,
		\ \ \quad  \forall\ \ct{u}{^{\alpha}}\ /\ \ct{n}{_{\mu}}\ct{u}{^{\mu}}=0\ ,\ct{u}{^\mu}\ct{u}{_{\mu}}=-1\ .
	\end{equation}	
This clearly imposes restrictions on $\ct{E}{_{\alpha\beta}}\prn{\vec{u}}$ and $\ct{B}{_{\alpha\beta}}\prn{\vec{u}}$ at the intersection of the initial hypersurface $\Sigma$ with $\scri$, but the former tensor field is determined by the ``timelike'' derivative of the second fundamental form $\ct{k}{_{\alpha\beta}}$ and the intrinsic Levi-Civita connection of $\Sigma$ determined by $\ct{h}{_{\alpha\beta}}$. Similarly, $\ct{B}{_{\alpha\beta}}\prn{\vec{u}}$ is determined by covariant derivatives within $\Sigma$ of the second fundamental form $\ct{k}{_{\alpha\beta}}$. Hence the criterion constraints the Cauchy data on $\Sigma\cap \scri$. The other piece of data, i.e., the boundary data given by $[\cts{g}{_{\alpha\beta}}]$, is also constrained by the radiation condition. To see this, recall that, as follows from \cref{thm:linearly-dependent}, the criterion establishes a linear relationship between  $\ct{C}{_{\alpha\beta}}$ and $\ct{D}{_{\alpha\beta}}$. Now one has to notice that $\ct{C}{_{\alpha\beta}}$ coincides up to a constant factor with the Cotton-York tensor of $\prn{\scri,\cts{g}{_{ab}}}$, hence \cref{crit:grav-rad} restricts also the conformal class of metrics on $\scri$, and its interplay with $D_{\alpha\beta}$.\\

 A very important remark is that Friedrich's construction depends on the choice of the initial hypersurface, but the set-up always assumes that $u^\alpha$ is orthogonal to $n^\alpha$ at the intersection of the initial hypersurface $\Sigma$ with $\scri$, in other words, that $u^\alpha$ is tangent to $\scri$ at that intersection. However, there are plenty of initial hypersurfaces leading to the same physical spacetime, each of them will have a different timelike normal $u^\alpha$, {\em but all of them} will be tangent to $\scri$ by construction. Since the presence of gravitational radiation is independent of the choice of the initial hypersurface $\Sigma$ ---as long as they lead to the same physical spacetime---, and thus independent of the timelike $u^\alpha$ normal to $\Sigma$, it is a must  to ask that \cref{eq:commutator} be satisfied for all $\ct{u}{^{\alpha}}$ tangent to $\scri$. This further and convincingly supports our formulation of \cref{crit:grav-rad}. \\
 
Moreover, there is an alternative and independent way of reaching the same conclusions. Recently \cite{Holzegel-Shao2023} ---see also \cite{Shao2024}---, it has been shown that, provided a conformally invariant geometric condition is satisfied, there is a one-to-one correspondence between the zeroth $g^{(0)}$ and $n$-th $g^{(n)}$ coefficients (in our case, $n=3$) of the Fefferman–Graham expansion \cite{Fefferman-Graham-1985,Fefferman-Graham-2011} and the solution $\ct{g}{_{\alpha\beta}}$ in a nieghbourhood domain of the conformal boundary $\scri$. Remarkably, this result includes dynamical space-times, hence gravitational radiation should be determined by the two pieces of data $\prn{g^{(0)},g^{(3)}}$. In the conformal setup, $g^{(0)}$ coincides with the metric on $\scri$ ---that determines the Cotton-York tensor, or $\ct{C}{_{\alpha\beta}}$---,  whereas $g^{(3)}$ essentially coincides with the tensor $\ct{D}{_{\alpha\beta}}$ of the rescaled Weyl tensor built with respect to the normal $\ct{n}{^{\alpha}}$ to $\scri$, that is to say, \eqref{eq:n-electric}. Following \cref{thm:linearly-dependent}, then, the no-radiation condition of \cref{crit:grav-rad} sets a precise and definite constraint between the two pieces of `holographic data'\footnote{In the context of the AdS/CFT correspondence, $g^{(n)}$ is typically known as the holographic stress-energy tensor of the boundary \cite{Balasubramanian1999}.} $\prn{g^{(0)},g^{(3)}}$. This was also the idea behind \cite{Ciambelli2024} , but their condition is just necessary as discussed in the previous Remark \ref{sabrina}.
	\subsection{New boundary conditions}
				
		The question of what boundary conditions are to be posed on $\scri$ from the viewpoint of the IBVP has its own interest in studies of AdS instabilities and on the definition of physically reasonable gravitating systems ---see the discussion in \cite{Friedrich2014}. Friedrich found a \emph{family} of covariant and conformally invariant conditions in his work \cite{Friedrich1995}. According to \cite{Fournodavlos-Smulevici-2023} there are also other type of conditions, called dissipative, that can be proved to lead to a well-posed mathematical problem, but such that geometric uniqueness cannot be ensured due to the {\em gauge dependence} of the boundary conditions. In this section we are going to provide new boundary conditions that are fully gauge independent and, besides, that can be related to the existence/absence of graviational radiation, and also to its in- or out-going character. 
		
		To be more specific, the general boundary condition proposed by Friedrich reads, in our notation, 
		$$
		\phi_4 - f_1 \phi_0 - f_2 \bar\phi_0 = f_3,  \hspace{4mm} |f_1|+|f_2| \leq 1\ ,
		$$
		where $f_1, f_2, f_3 $ are complex functions on $\scri$.
Unfortunately, in its general form such a condition is not covariant in the sense that it depends on the choice of null tetrad. Thus, in order to get covariant boundary conditions, one must restrict them \cite{Friedrich2009,Friedrich2014} to a simpler version
$$
\phi_4 -\bar\phi_0 =f_3.
$$
Gauge independence can be accomplished by adding
\begin{equation}\label{FriBC}
\phi_3 -\bar \phi_1 =0, \hspace{3mm} \phi_2 -\bar \phi_2 =0 \hspace{5mm} \mbox{on} \scri\cap \Sigma. 
\end{equation}
Obviously, the case with $f_3=0$ is included in our no-radiation criterion \eqref{eq:rad-cond-tetrad-coplanar} as it is clear from its version \eqref{linear1}. Notice that there is a subtle point here: our condition \eqref{eq:rad-cond-tetrad-coplanar} can be required in an IBVP for a given initial $\Sigma$, but to ensure absence of gravitational radiation one needs that \eqref{eq:rad-cond-tetrad-coplanar} holds for all timelike $l^\alpha + k^\alpha$ tangent to $\scri$, leading to \eqref{linear1}, but also to \eqref{linear} and \eqref{linear2}. And this is the role played by the corner conditions in \eqref{FriBC}.

As shown in \cite{Friedrich2009,Friedrich2014}, \eqref{FriBC} entail $C_{\alpha\beta}|_{\scri\cap \Sigma} =0$ which, together with the divergence-free property of $C_{\alpha\beta}$ (a property of the given conformal structure on $\scri$) lead to
			\begin{equation}\label{eq:conformal-flatness}
				\ct{C}{_{\alpha\beta}}=0\ .
			\end{equation}
			That is, conformal flatness of $\scri$. Observe that this is clearly gauge and tetrad independent, and indeed compatible with \eqref{linear}--\eqref{linear2}:  it is just the case with $\beta =0 $ \footnote{Thus, conformal flatness implies absence of gravitational radiation at $\scri$. Also, for any spherically symmetric metric the Petrov type is D, leading to a rescaled Weyl tensor of Petrov type D or zero at $\scri$. In the former case, the two PNDs are orthogonal to the orbits of the rotation group (the round spheres) and thus the normal to scri will be coplanar with them, leading to absence of radiation as it must be.}. This condition on the geometry of the conformal boundary was termed as ``reflective'' many years before by Hawking \cite{Hawking1983}. Since then, it has been standardised in the literature and argued to be a natural requirement for the Anti-de Sitter scenario \cite{Ashtekar-Magnon1984,Ashtekar-Das2000}. However, the question of finding other sort of \emph{reasonable physical} boundary conditions in a gauge-invariant way (that is, not depending on the choice of coordinates nor on the conformal gauge) is an open question \cite{Fournodavlos-Smulevici-2023}. 
			
			 In that sense, one can use criterion \ref{crit:grav-rad}, as the conformal flatness condition \eqref{eq:conformal-flatness} is just a particular class within the \emph{broader family of asymptotically non-radiating space-times} given by \cref{eq:condition-prop-covariant}. Thus, we have found the most general form of boundary conditions leading to absence of radiation traversing $\scri$, as follows from Theorem \ref{thm:linearly-dependent}. They are given in intrinsic form by \eqref{eq:condition-prop-covariant}, and they contain two arbitrary functions on $\scri$, $\beta$ and $\gamma$, that can be given arbitrarily. Their version in terms of standard asymptotic Weyl rescaled scalars were given in \eqref{linear}--\eqref{linear2} (recall $\alpha: \beta +i\gamma$):
			 \begin{eqnarray*}
			 \alpha\bar\phi_4 =-\bar\alpha \phi_0,\\
			 \alpha \bar\phi_1 =-\bar\alpha \phi_3,\\
			 \alpha\phi_2 =- \bar\alpha\bar\phi_2 .
			 \end{eqnarray*}
One can take them as conditions on $\scri$, or some of them as corner condition on $\scri\cap \Sigma$ that will then be propagated along $\scri$. Notice. however, that the specific form of $\alpha$ might be restricted by the particular energy-momentum content of the physical spacetime that one wishes to construct from the IBVP, cf. remark \ref{remarkgradient}.

			 Similarly, other kind of covariant boundary conditions can be explored by allowing for \emph{the presence of gravitational radiation} by using criterion \ref{crit:no-outgoing-gw}, distinguishing between incoming and outgoing gravitational radiation using \cref{thm:outgoing-gw}. The basic idea now is to keep property \eqref{property} for all timelike $u^\alpha$ tangent to $\scri$. There may be many choices here, but we have identified two very simple examples, given by \eqref{example} and \eqref{example1}.  We can rewrite these conditions in terms of rescaled Weyl scalars by using an appropriate null tetrad. The best choice is to select a unit timelike $u^\alpha$ tangent to $\scri$ but orthogonal to the vector field $p^\alpha$ appearing in \eqref{example} and \eqref{example1}. Then we define the real null vectors in the tetrad in the usual manner, i.e. \cref{eq:ellcoplanar,eq:kcoplanar}. By construction then $p^\alpha$ is orthogonal to both $\ell^\alpha$ and $k^\alpha$, and can be used as the real part of the complex null $m^\alpha=(p^\alpha -i q^\alpha)/\sqrt{2}$. Here $q^\alpha$ is another unit spacelike vector, tangent to $\scri$ but orthogonal to $p^\alpha$ and $u^\alpha$. Notice that
			 \begin{eqnarray*}
			 n_\alpha L^\alpha =0 \Longrightarrow \hspace{5mm} \ell_\alpha L^\alpha = k_\alpha L^\alpha ,\\
			 L_\alpha L^\alpha =0 \Longrightarrow \hspace{5mm} L_\alpha q^\alpha = \sqrt{2} L_\alpha k^\alpha\ ,
			 \end{eqnarray*}
where in the last one a sign has been fixed by choice of orientation of $q^\alpha$. So, only one of these scalar products is independent, and thus it is convenient to give it a short name:
$$
Y:= k_\alpha L^\alpha <0.
$$
Now, by following the same calculations used in \eqref{eq:condition-propto}, \eqref{lin0}, \eqref{lin2} and \eqref{lin1}, conditions \eqref{example} lead to
\begin{eqnarray*}
\phi_2 &=&X  Y^2 ,\\
\phi_3 &=& i\frac{Y}{\sqrt{2}} \left( 2YX +Z\right),\\
\phi_1 &=& - i\frac{Y}{\sqrt{2}} \left( 2YX -Z\right),\\
\phi_0 &=& -Y(YX+\sqrt{2} Z),\\
\phi_4 &=& Y( \sqrt{2} Z-YX) .
\end{eqnarray*}
These can be used as boundary conditions for the case of no-incoming radiation through $\scri$ ($XZ\geq 0$) or no-outgoing radiation through $\scri$ ($XZ\leq0$). Observe that $\phi_2$ is real and positive, $\phi_1$ and $\phi_3$ are purely imaginary, and $\phi_0$ and $\phi_4$ are real. Note also that these conditions can be rewritten avoiding $Z$, $X$ and $Y$ as follows
\begin{eqnarray}
\phi_0+\phi_4 = -2\phi_2 =-2\bar\phi_2 =\frac{i}{\sqrt{2}}(\phi_3 -\phi_1) , \label{phis1}\\
\phi_4-\phi_0 = -2 i (\phi_1 +\phi_3) \label{phis2}.
\end{eqnarray}
All the elements of the first line \eqref{phis1} have the sign of $-X$, while those on the second line \eqref{phis2} have the sign of $-Z$. And this is how one controls the sign of $XZ$, that defines the absence of incoming or of outgoing radiation, for instance
\begin{equation}\label{sign}
\phi_4^2 -\phi_0^2 \geq 0 \hspace{5mm} (\mbox{or }  \leq 0).
\end{equation}

Similarly, conditions \eqref{example1} lead to
\begin{eqnarray*}
\phi_2 &=& i X Y^2 ,\\
\phi_3 &=& \frac{Y}{\sqrt{2}} \left( 2YX +Z\right),\\
\phi_1 &=& \frac{Y}{\sqrt{2}} \left( Z- 2YX \right),\\
\phi_0 &=& -i Y(\sqrt{2} Z-YX),\\
\phi_4 &=& i Y( \sqrt{2} Z+YX) .
\end{eqnarray*}
Now, $\phi_2$, $\phi_0$ and $\phi_4$ are purely imaginary, while $\phi_1$ and $\phi_3$ are real. Eliminating again $X$, $Y$ and $Z$ the conditions read
\begin{eqnarray}
i(\phi_0+\phi_4) = 2i\phi_2 =-2i\bar\phi_2 =\frac{1}{\sqrt{2}}(\phi_1 -\phi_3)   \label{phis3}\\
i(\phi_0-\phi_4) = 2  (\phi_1 +\phi_3) \label{phis4} .
\end{eqnarray}
Now,  the elements of the first line \eqref{phis3} have the sign of $-X$, while those on the second line \eqref{phis4} have the sign of $Z$. Thus we must add
\begin{equation}\label{sign1}
\phi_1^2 -\phi_3^2 \leq 0 \hspace{5mm} (\mbox{or }  \geq 0).
\end{equation}
As far as we know, these are new boundary conditions and they come with a clear interpretation: gravitational radiation only arriving at $\scri$, or gravitational radiation only entering the physical spacetime from $\scri$.
				
\section{Discussion}
This paper completes our program of characterizing the existence of gravitational radiation for arbitrary values of the cosmological constant $\Lambda$, based on the tidal properties of the gravitational field. Our driving force has been to consider the {\em flux of tidal energy} ---measured by {\em observers geometrically selected} by the structure of infinity in the conformal completion--- as the fundamental physical phenomenon signaling the existence of gravitational radiation traversing $\scri$. Thus, the underlying idea is simple, but has proven to be very powerful. Our characterization is equivalent to the news tensor characterization when $\Lambda =0$ \cite{Fernandez-Alvarez-Senovilla2020a,Fernandez-Alvarez-Senovilla2022a}, and it has given the correct results when $\Lambda >0$ \cite{Fernandez-Alvarez-Senovilla2020b,Fernandez-Alvarez-Senovilla2022b}, in particular showing  that ---within the general family of type D exact solutions representing (pairs of) black holes--- only {\em accelerated} ones radiate \cite{Fernandez-Alvarez-Podolsky-Senovilla2024}. In all cases, the basic physical quantity we have used is the asymptotic super-Poynting vector, a vector field that depends on the observer and is built with the Bel-Robinson tensor (i.e. the tidal energy-momentum tensor) of the rescaled Weyl tensor in analogy to the Poynting vector of electromagnetism.\\

In this paper, the case of $\Lambda <0$ has been settled. The main problem in comparison with the $\Lambda \geq 0$ cases is the absence of a privileged causal observer at $\scri$, given that now $\scri$ is itself a Lorentzian 3-dimensional geometry. This obstacle is resolved by simply asking that {\em all} possible observers within $\scri$ do not see flux of tidal energy. This was substantiated in criterion \ref{crit:grav-rad}, and later proven to be equivalent to the covariant condition \eqref{eq:condition-prop-covariant} requiring linear dependency of $\scri$'s Cotton-York tensor and the so-called holographic stress tensor, both built using only the rescaled Weyl tensor at $\scri$ and the unit normal there. Condition \eqref{eq:condition-prop-covariant} is the most general prescription compatible with absence of gravitational radiation at $\scri$, it can be applied even in the presence of matter subject to \eqref{T->0}, and it generalizes the standard ``reflective'' cases studied in the literature, as has been largely explained in section \ref{sec:IBVP}. In particular, our results fit perfectly well with the Friedrich analysis of the initial-boundary value problem for vacuum spacetimes with $\Lambda$ as well as with the mathematical results of the Fefferman-Graham expansions for timelike conformal boundaries. They may have some interesting consequences in the program of ``celestial holography'', trying to apply the holographic principle to the boundary of conformal completions of the physical spacetime, as has been discussed in \cite{Ciambelli2024} ---but notice remark \ref{sabrina}.\\

When radiation is absent at infinity, the structure of the rescaled Weyl tensor geometry at $\scri$  has been fully described in theorem \ref{thm:nograv-algebraic}. Essentially, the spacelike normal to $\scri$ must either be orthogonal to the principal null directions, or must be coplanar with a pair of them (or, in some situations, both of these). Such kind of alignments occur for all signs of $\Lambda$, in each case with its own peculiarities. For $\Lambda =0$ the normal is itself null, and absence of radiation requires it to be a multiple PND (and thus orthogonal to itself) \cite{Fernandez-Alvarez-Senovilla2020a,Fernandez-Alvarez-Senovilla2022a}, so that $d_{\alpha\beta\mu}{}^\nu$ is algebraically special at $\scri$. When $\Lambda >0$ , the timelike normal is a principal (timelike!) direction of the rescaled Weyl tensor \cite{Fernandez-Alvarez-Senovilla2020b,Fernandez-Alvarez-Senovilla2022b,Senovilla2022}, and this can only happen if $d_{\alpha\beta\mu}{}^\nu$ is type D or I at $\scri$. As we have shown herein, the richer case arises when $\Lambda<0$, as all possible Petrov types are allowed for $d_{\alpha\beta\mu}{}^\nu$ at $\scri$, as long as the PNDs organize themselves in an appropriate way with respect to the spacelike normal $n^\alpha$.\\

The cases with only outgoing radiation escaping from the physical spacetime (arriving at $\scri$) or only with ingoing radiation entering the physical spacetime (from $\scri$) have also been fully characterized in criterion \ref{crit:no-outgoing-gw}, and also in a covariant manner \eqref{property}, and some explicit examples \eqref{example} and \eqref{example1} have been provided. For the benefit of some readers we have written all proposed boundary data at $\scri$ explicitly in terms of the rescaled Weyl scalars, both for absence of radiation and cases with outgoing or incoming radiation. 
	\subsection*{Acknowledgments}
				 FFA was partly supported by the Grant Margarita Salas MARSA22/20 (Spanish Ministry of Universities and European Union), financed by European Union -- Next Generation EU.
				 JMMS is supported by the Basque Government grant number IT1628-22, and by Grant PID2021-123226NB-I00 funded by the Spanish MCIN/AEI/10.13039/501100011033 together with ``ERDF A way of making Europe''.

\appendix
\section{Orthogonal splitting of a Bel-Robinson tensor with respect to a unit spacelike vector}\label{App:decom}
The following formulas are needed in section \ref{sec:implications} but, given that they are general and may have an independent interest, we present them here for a general Weyl tensor candidate, i.e., any tensor with the symmetry and trace properties of the Weyl tensor.

Thus, let $W_{\alpha\beta\lambda}{}^\mu $ be any Weyl tensor candidate, so that
$$
W_{\alpha\beta\lambda\mu} =W_{[\alpha\beta][\lambda\mu]} = W_{\lambda\mu\alpha\beta} , \hspace{3mm} 
W_{\alpha[\beta\lambda\mu]}=0, \hspace{3mm} W_{\alpha\mu\lambda}{}^\mu =0.
$$
Given any such $W_{\alpha\beta\lambda}{}^\mu $ and any unit spacelike vector $n^\mu$ ($n^\mu n_\mu =1$), one can define two tensors orthogonal to $n^\mu$ that fully determine $W_{\alpha\beta\lambda}{}^\mu $, given by
\begin{align}
	 	    \ct{\C}{_{\alpha\beta}} &\defeqs \ct{n}{^{\mu}}\ct{n}{^{\nu}}\ctr{^*}{W}{_{\alpha\mu\beta\nu}} \ ,
		    \\
	 	    \ct{\D}{_{\alpha\beta}} &\defeqs \ct{n}{^{\mu}}\ct{n}{^{\nu}}\ct{W}{_{\alpha\mu\beta\nu}}\ .
	 	\end{align}
		Observe the analogy of these tensors with the typical electric and magnetic parts of a Weyl tensor, which are built with a unit timelike $u^\mu$ instead of $n^\mu$. Still, the interpretation of these two tensors is quite different, for instance notice that they are not spacelike tensors, and thus their squares $\ct{\C}{_{\alpha\beta}}\ct{\C}{^{\alpha\beta}}$ and $\ct{\D}{_{\alpha\beta}} \ct{\D}{^{\alpha\beta}} $ are not non-negative definite, nor they can be generically diagonalized as they live in a Lorentzian space. The properties of $\ct{\C}{_{\alpha\beta}} $ and $\ct{\D}{_{\alpha\beta}} $ are
		$$\D_{\alpha\beta} = \D_{(\alpha\beta)}, \hspace{3mm} \D_{\alpha\beta} n^\beta =0, \hspace{3mm} \D^\mu{}_\mu =0, \hspace{3mm} \C_{\alpha\beta} = \C_{(\alpha\beta)}, \hspace{3mm} \C_{\alpha\beta} n^\beta =0, \hspace{3mm} \C^\mu{}_\mu =0.$$
These two tensors fully determine $W_{\alpha\beta\lambda}{}^\mu $ (together with $n^\mu$) by means of the following explicit expression
$$
W_{\alpha\beta\lambda\mu} =(g_{\alpha\beta\rho\sigma} g_{\lambda\mu\tau\nu} -\eta_{\alpha\beta\rho\sigma} \eta_{\lambda\mu\tau\nu})n^\rho n^\tau \D^{\sigma\nu} - (g_{\alpha\beta\rho\sigma} \eta_{\lambda\mu\tau\nu} -\eta_{\alpha\beta\rho\sigma} g_{\lambda\mu\tau\nu})n^\rho n^\tau \C^{\sigma\nu} 
$$
where 
$$g_{\alpha\beta\lambda\mu}= g_{\alpha\lambda} g_{\beta\mu} -g_{\alpha\mu} g_{\beta\lambda} .
$$
It follows that
\begin{equation}\label{nW}
n^\alpha W_{\alpha\beta\lambda\mu} = 2\D_{\beta[\mu} n_{\lambda]} -\C_{\beta}{}^\rho \epsilon_{\rho\lambda\mu} 
\end{equation} 
where 
\begin{equation}\label{neta}
\epsilon_{\beta\mu\nu} := n^\alpha \eta_{\alpha\beta\mu\nu} , \hspace{3mm} n^\beta \epsilon_{\beta\mu\nu} =0.
\end{equation}
Define the corresponding Bel-Robinson tensor \cite{Bel1958, Senovilla2000}
$$
 			  \ct{\B}{_{\alpha\beta\gamma\delta}}\defeq \ct{W}{_{\alpha\mu\gamma}^\nu}\,\ct{W}{_{\delta\nu\beta}^\mu}
 	     +  \ctr{^*}{W}{_{\alpha\mu\gamma}^\nu}\,\ctr{^*}{W}{_{\delta\nu\beta}^\mu} =
	     \ct{W}{_{\alpha\mu\gamma}^\nu}\,\ct{W}{_{\delta\nu\beta}^\mu}+\ct{W}{_{\alpha\mu\delta}^\nu}\,\ct{W}{_{\gamma\nu\beta}^\mu}-\frac{1}{8}\ct{g}{_{\alpha\beta}}\ct{g}{_{\gamma\delta}} \ct{W}{_{\rho\mu\sigma}^\nu}\,\ct{W}{^{\rho\mu\sigma}_\nu} \, .
$$
This is a fully symmetric and traceless tensor. Its orthogonal splitting with respect to $n^\alpha$ reads (see  \cite{Alfonso2008} for the corresponding splitting with respect to a timelike vector)
\begin{equation}\label{split}
\ct{\B}{_{\alpha\beta\gamma\delta}}= w\,  n_\alpha n_\beta n_\gamma n_\delta +4 q_{(\alpha} n_\beta n_\gamma n_{\delta)}+6\tau_{(\alpha\beta}n_\gamma n_{\delta)}+ 4 q_{(\alpha\beta\gamma} n_{\delta)} + \tau_{\alpha\beta\gamma\delta}
\end{equation}
where $q_\alpha, \tau_{\alpha\beta}, q_{\alpha\beta\gamma}$ and $\tau_{\alpha\beta\gamma\delta}$ are fully symmetric tensors, all of them orthogonal to $n^\alpha$, and besides
$$
\tau^\rho{}_{\rho\gamma\delta}=-\tau_{\gamma\delta}, \hspace{3mm} \tau^\rho{}_\rho =-w, \hspace{3mm} q^\rho{}_{\rho\gamma} = -q_{\gamma}.
$$
By using \eqref{nW} it is easy to obtain the expression of $q_{\alpha}$ in terms of $\ct{\C}{_{\alpha\beta}} $ and $\ct{\D}{_{\alpha\beta}} $:
\begin{equation}\label{q}
q_{\mu} \defeq \ct{\B}{_{\alpha\beta\gamma\rho}} n^\alpha n^\beta n^\gamma h^\rho_\mu  = 2\epsilon_{\mu\rho\sigma} \C^{\tau \rho} \D^\sigma{}_\tau .
\end{equation}
Here, $h^\rho{}_\mu$ is the projector orthogonal to $n^\alpha$
\begin{equation}\label{projector}
h^\rho{}_\mu \defeq \delta^\rho_\mu -n^\rho n_\mu .
\end{equation}
In simpler words, $q_\alpha$ is related to the commutator of the $3\times 3$ matrices defined by $\ct{\C}{_{\alpha\beta}} $ and $\ct{\D}{_{\alpha\beta}} $. 

A similar but longer calculation provides the expression for $q_{\alpha\beta\gamma}$ in terms of $\ct{\C}{_{\alpha\beta}} $ and $\ct{\D}{_{\alpha\beta}} $:
\begin{equation}\label{Q}
q_{\lambda\mu\nu} \defeq \ct{\B}{_{\alpha\beta\gamma\delta}} n^\alpha h^\beta{}_\lambda h^\gamma{}_\mu h^\delta{}_\nu = -q_\lambda h_{\mu\nu} +2\epsilon_{\lambda}{}^{\rho\sigma} (\D_{\mu\sigma} \C_{\nu\rho} +\D_{\nu\sigma} \C_{\mu\rho} ).
\end{equation}

\section{A general result on superenergy. (Derivation of \cref{eq:formula270})}\label{sec:appendix}
The main steps to derive \cref{eq:formula270} are presented in this appendix. The formula is general, for any Weyl tensor candidate $\ct{d}{_{\alpha\beta\gamma}^\delta}$, and pointwise, but the same notation as in the main text is going to be used for clarity. \\
Let $\ct{u}{^\alpha}$ be a unit, timelike, future-oriented vector, and consider the tetrads presented in \cref{fig:set1} ---recall that each future-oriented null vector with subindex $i=1,2,3,4$ is aligned with a different PND of a Weyl tensor candidate. In the starting basis $\prn{{k}{_{1}^\alpha},{k}{_{2}^\alpha}, \ct{m}{^\alpha}}$, write $\ct{u}{^{\alpha}}$ as
\begin{equation}\label{eq:u-components}
	\ct{u}{^\alpha}=a{k}{_1^{\alpha}}+b{k}{_2^\alpha}+d\ct{m}{^\alpha}+\bar{d}\ct{\bar{m}}{^\alpha}\ .
\end{equation}
The components must satisfy
\begin{equation}\label{eq:ab-dd}
ab-d\bar{d}=\frac{1}{2}\ ,\quad a,b>0
\end{equation}
Next, consider a unit space-like vector $\ct{n}{^\alpha}$ orthogonal to $\ct{u}{^{\alpha}}$. Then, contracting  \cref{eq:u-components} with $\ct{n}{_{\alpha}}$, 
	\begin{equation}
	aB_1+bB_2+dD+\bar{d}\bar{D}=0\ ,\label{eq:u-n-orthogonal}
	\end{equation}
where $B_1\defeq\ct{n}{_{\mu}}{k}{_1^\mu}$, $B_2 \defeq \ct{n}{_{\mu}}{k}{_2{^\mu}}$, $D\defeq\ct{n}{_{\mu}}\ct{m}{^\mu}$. \\
The super-Poynting vector $\ct{\P}{^{\alpha}}$ associated to any Weyl tensor candidate and vector $\ct{u}{^\alpha}$ was written in a null basis in an appendix of \cite{Fernandez-Alvarez-Senovilla2022b}. Using that formula, one has (`c.c.' stands for `complex conjugate')
	\begin{align}
	   -\frac{1}{4} \ct{n}{_{\mu}} \cts{\P}{^\mu}\prn{\vec{u}}& =\phi_{1}\bar{\phi}_{1}\brkt{-4B_1a^3-6\prn{ab+d\bar{d}}aB_2}\nonumber\\
	     +& \phi_{1}\bar{\phi}_{2}\brkt{-6a^2dB_1-3\prn{3ab+d\bar{d}}dB_2+3\prn{ab+d\bar{d}}a\bar{D}}+c.c\nonumber\\
	     +&\phi_1\bar{\phi}_{3}\brkt{-4ad^2B_1-4bd^2B_2+2\prn{3ab+d\bar{d}}d\bar{D}}+c.c.\nonumber\\
	     -& \phi_2\bar{\phi}_{2}9\prn{ab+d\bar{d}}\prn{aB_1+bB_2}\nonumber\\
	     +&\phi_{2}\bar{\phi}_{3}\brkt{-3\prn{3ab+d\bar{d}}dB_{1}-6b^2dB_2+3\prn{ab+d\bar{d}}b\bar{D}}+c.c.\nonumber\\
	     +&\phi_{3}\bar{\phi}_{3}\brkt{-6\prn{ab+d\bar{d}}bB_1-4b^3B_2}\ .\label{eq:nP-original}
	\end{align}
Next, one considers the null rotations indicated in \cref{fig:set1}. These are very well known relations \cite{Stephani2003}. For illustration, the formulae transforming $\prn{{k}{_{1}^\alpha},{k}{_{2}^\alpha}, \ct{m}{^\alpha}}$ into $\prn{\tilde{k}{_{3}^\alpha},{k}{_{2}^\alpha}, \tilde{m}{^\alpha}}$ are presented:
	\begin{align}
	    \tilde{k}^\alpha_3&=k_1^\alpha+\tilde{c}m^\alpha+\bar{\tilde{c}}\bar{m}^\alpha+\tilde{c}\bar{\tilde{c}}k^\alpha_{2}\ , \\
	    \tilde{m}^\alpha &= m^\alpha+\bar{\tilde{c}}k_2^\alpha\ .
	\end{align}
In the new basis, define the components of $\ct{u}{^\alpha}$ as
	\begin{equation}\label{eq:utilde}
		\ct{u}{^\alpha}=\tilde{a}\tilde{k}_{3}^\alpha+\tilde{b}k_{2}^\alpha+\tilde{d}\tilde{m}^\alpha+\bar{\tilde{d}}\bar{\tilde{m}}^\alpha\ ,
	\end{equation}
and it can be checked that
	\begin{align}
	    \tilde{a} &= a \ ,\label{eq:tildea}\\
	    \tilde{b} &= b+\tilde{c}\bar{\tilde{c}}a-\prn{\tilde{c}\bar{d}+\bar{\tilde{c}}d}\ ,\label{eq:tildeb}\\
	    \tilde{d} &= d-a\tilde{c}\ .\label{eq:tilded}
	\end{align}
The transformation between the scalars of the Weyl-tensor candidate in the old $\ct{\phi}{_{i}}$ and new $\ct{\tilde{\phi}}{_{i}}$ basis read
	\begin{align}
		\tilde{\phi_{0}}&=4\bar{\tilde{c}}\phi_{1}+6\bar{\tilde{c}}^2\phi_{2}+4\bar{\tilde{c}}^3\phi_{3}=0,\\
		\tilde{\phi}_{1}&=\phi_{1}+ 3\bar{\tilde{c}}\phi_{2}+3\bar{\tilde{c}}^2\phi_{3}\\
		\tilde{\phi}_{2}&=\phi_{2}+2\bar{\tilde{c}}\phi_{3}\\
		\tilde{\phi}_{3}&=\phi_{3}\\
		\tilde{\phi}_{4}&=\phi_{4}=0\ .
			\end{align}
Recall that the null directions are aligned with different PNDs ($i=1,2,3,4$), hence $\phi_{0}=0$, $\phi_{4}=0$ in all the bases considered.  One can then write the formulae for each of the null rotations indicated in \cref{fig:set1} by adding the corresponding decorations (hats, tildes, etc). Importantly, the four null rotations are not independent, as one can check that
	\begin{equation}\label{eq:c-inverse}
		\bar{c}^{\prime}=\frac{1}{\tilde{c}}\ ,\quad\bar{\hat{c}}=\frac{1}{{\breve{c}}}\ .
	\end{equation}
Also,
	\begin{equation}
	\breve{k}_{4}^\alpha=\breve{c}\bar{\breve{c}}\hat{k}_{4}^\alpha\ ,\quad k_{3}^{\prime\alpha}=c'\bar{c}'k^\alpha_{3}\ ,
	\end{equation}
	\begin{equation}
	\breve{b}=\breve{c}\bar{\breve{c}}\hat{a}\ ,\quad a'=c'\bar{c}'\tilde{b}\ .
	\end{equation}
The products $\tilde{B_3}\defeq \ct{n}{_{\alpha}}\tilde{k}^\alpha_{3}$,${B^\prime_3}\defeq \ct{n}{_{\alpha}}{k^\prime}^\alpha_{3}$, $\breve{B}_4\defeq \ct{n}{_{\alpha}}\breve{k}_{4}^\alpha$, $\hat{B}_4\defeq \ct{n}{_{\alpha}}\hat{k}_{4}^\alpha$ fulfil
	\begin{align}
		\tilde{B}_3&=B_1+\tilde{c}D+\bar{\tilde{c}}\bar{D}+\tilde{c}\bar{\tilde{c}}B_2\ ,\label{eq:B3-B1B2}\\
		\breve{B}_{4}&=B_1+\breve{c}D+\bar{\tilde{c}}\bar{D}+\breve{c}\bar{\breve{c}}B_2\ \label{eq:B4-B1B2},\\
		\tilde{B}_3&=\tilde{c}\bar{\tilde{c}}B^\prime_3\ ,\\
		\breve{B}_{4}&=\breve{c}\bar{\breve{c}}\hat{B}_{4}\ .\label{eq:B4}
	\end{align} 	

Using the relations between the $\phi_i$ in the different bases, it also follows that
	\begin{align}
		\phi_{2}&=-\frac{2}{3}\phi_{3}\prn{\bar{\breve{c}}+\bar{\tilde{c}}}=-\frac{2}{3}\phi_{1}\prn{\hat{c}+c'}\ ,\label{eq:phis-relation-1}\\
		\phi_{3}&=c^\prime \hat{c}\phi_{1}\ ,\quad  \phi_{1}=\bar{\breve{c}}\bar{\tilde{c}}\phi_{3}\ .\label{eq:phis-relation-2}
	\end{align}
The idea is to express \cref{eq:nP-original} in terms of $B$'s solely, removing all complex $D$'s. For that, one has to substitute $\phi_1$ and $\phi_2$ in terms of $\phi_3$ using \cref{eq:phis-relation-1,eq:phis-relation-2}. After doing so, one encounters terms that read
	\begin{align}
		2\prn{\bar{\breve{c}}\bar{\tilde{c}}d\bar{D}+\breve{c}\tilde{c}\bar{d}D}=\prn{d\bar{\tilde{c}}+\bar{d}\tilde{c}}\prn{\bar{\breve{c}}\bar{D}+\breve{c}D}+\prn{d\bar{\breve{c}}+\bar{d}\breve{c}}\prn{\bar{\tilde{c}}\bar{D}+\tilde{c}D}-\prn{\tilde{c}\bar{\breve{c}}+\bar{\tilde{c}}{\breve{c}}}\prn{dD+\bar{d}\bar{D}}\ ,
	\end{align}
and using \cref{eq:B3-B1B2,eq:B4-B1B2,eq:u-n-orthogonal} they can be expressed as
	\begin{align}
	2\prn{\bar{\breve{c}}\bar{\tilde{c}}d\bar{D}+\breve{c}\tilde{c}\bar{d}D}&=\prn{d\bar{\tilde{c}}+\bar{d}\tilde{c}}\prn{\breve{B}_{4}-B_{1}-\breve{c}\bar{\breve{c}}B_{2}}\nonumber\\
	&+\prn{d\bar{\breve{c}}+\bar{d}\breve{c}}\prn{\tilde{B}_{3}-B_{1}-\tilde{c}\bar{\tilde{c}}B_2}\nonumber\\
	&+\prn{\tilde{c}\bar{\breve{c}}+\bar{\tilde{c}}{\breve{c}}}\prn{aB_1+bB_2}\ .
	\end{align}
After a long rearrangement of terms using this kind of substitutions, one obtains 
	\begin{equation}\label{eq:rearraengement}
	 -\frac{1}{4}\ct{n}{_{\mu}} \cts{\P}{^\mu}\prn{\vec{u}}=\phi_{3}\bar{\phi}_{3}\prn{J_{1}B_1+ J_{2}B_2+ J_3\tilde{B}_3+J_4\breve{B}_4}\ ,
	\end{equation}
where
	\begin{align}
	    J_{1} &=- a \Bigg[4\prn{\breve{c}\bar{d}+\bar{\breve{c}}d}\prn{\tilde{c}\bar{d}+\bar{\tilde{c}}d}-4a\breve{c}\bar{\breve{c}}\prn{\tilde{c}\bar{d}+\bar{\tilde{c}}d}-4a\tilde{c}\bar{\tilde{c}}\prn{\bar{\breve{c}}d+\breve{c}\bar{d}} \nonumber\\
	    +&4a^2\breve{c}\bar{\breve{c}}\tilde{c}\bar{\tilde{c}}+\prn{ab+3d\bar{d}}\prn{\breve{c}\bar{\breve{c}}+\tilde{c}\bar{\tilde{c}}}+\prn{ab-d\bar{d}}\prn{\tilde{c}+\breve{c}}\prn{\bar{\tilde{c}}+\bar{\breve{c}}}\Bigg] \nonumber\\
	    + & \prn{3ab+d\bar{d}}\prn{\tilde{c}\bar{d}+\bar{\tilde{c}}d+\breve{c}\bar{d}+\bar{\breve{c}}d}-2b\prn{ab+d\bar{d}}\ ,\label{eq:J1}\\
	    J_{2}  &= b\Bigg[-2\prn{ab+d\bar{d}}\prn{\tilde{c}\bar{\tilde{c}}+\breve{c}\bar{\breve{c}}}-\prn{\bar{\breve{c}}\tilde{c}+\breve{c}\bar{\tilde{c}}}\prn{ab-d\bar{d}}\nonumber \\
	    + &4b\prn{\tilde{c}\bar{d}+\bar{\tilde{c}}d+\breve{c}\bar{d}+\bar{\breve{c}}d}- 4\prn{\tilde{c}\bar{d}+\bar{\tilde{c}}d}\prn{\breve{c}\bar{d}+\bar{\breve{c}}d}-4b^2\Bigg] \nonumber\\
	    +&\prn{3ab+d\bar{d}}\brkt{\breve{c}\bar{\breve{c}}\prn{\tilde{c}\bar{d}+\bar{\tilde{c}}d}+\tilde{c}\bar{\tilde{c}}\prn{
	    \breve{c}\bar{d}+\bar{\breve{c}}d}}-2a\breve{c}\bar{\breve{c}}\tilde{c}\bar{\tilde{c}}\prn{ab+d\bar{d}}\ ,\label{eq:J2}\\
	    J_3&=\prn{3ab+d\bar{d}}\prn{\breve{c}\bar{d}+\bar{\breve{c}}d}-2a\breve{c}\bar{\breve{c}}\prn{ab+d\bar{d}}-2b\prn{ab+d\bar{d}}\ ,\label{eq:J3}\\
	    J_4&=\prn{3ab+d\bar{d}}\prn{\tilde{c}\bar{d}+\bar{\tilde{c}}d}-2a\tilde{c}\bar{\tilde{c}}\prn{ab+d\bar{d}}-2b\prn{ab+d\bar{d}}\ .\label{eq:J4}
	\end{align}
	$J_3$ and $J_4$ can be rewritten  in terms of the components of $\ct{u}{^{\alpha}}$ in the different basis of \cref{fig:set1} as\footnote{This manifestly positive expression is not unique; there exist other possible combinations of products of (decorated) $ab$ and $d\bar{d}$ with different coefficients. For convenience and brevity these are not presented here.}
		\begin{align}
		    -J_3 &= a\breve{c}\bar{\breve{c}}\prn{\frac{1}{3}\hat{a}\hat{b}+\hat{d}\bar{\hat{d}}}+b\prn{\frac{1}{3}\breve{a}\breve{b}+\breve{d}\bar{\breve{d}}}+\breve{b}\prn{\frac{1}{3}ab+d\bar{d}}\ , \label{eq:J3p}\\
		    -J_4 &=a\tilde{c}\bar{\tilde{c}}\prn{\frac{1}{3}a'b'+d'\bar{d}'}+b\prn{\frac{1}{3}\tilde{a}\tilde{b}+\tilde{d}\bar{\tilde{d}}}+\tilde{b}\prn{\frac{1}{3}ab+d\bar{d}}\ . \label{eq:J4p}
		\end{align} 
	Observe that the r.-h.s. of these equations has a positive sign. This is precisely the desired property. The expressions in \cref{eq:J1,eq:J2} for $J_1$ and $J_2$ look more complicated though. They can be manipulated in a similar way to the other two, giving
		\begin{align}
			-J_{1}&=\breve{b}\prn{\tilde{a}\tilde{b}+\tilde{d}\bar{\tilde{d}}}+\tilde{b}\prn{\breve{a}\breve{b}+\breve{d}\bar{\breve{d}}}-\frac{1}{2}a\prn{\tilde{c}-\breve{c}}\prn{\bar{\tilde{c}}-\bar{\breve{c}}}\ ,\label{eq:J1-step}\\
			-J_{2}&=\tilde{c}\bar{\tilde{c}}\breve{c}\bar{\breve{c}}\brkt{\hat{a}\prn{a'b'+d'\bar{d}'}+a'\prn{\hat{a}\hat{b}+\hat{d}\bar{\hat{d}}}}-\frac{b}{2}\prn{\tilde{c}-\breve{c}}\prn{\bar{\tilde{c}}-\bar{\breve{c}}}\ .\label{eq:J2-step}
		\end{align}
	 This is so because one has started the null rotations departing from the central basis in \cref{fig:set1}, i.e., $\cbrkt{\ct{k_1}{^\alpha},\ct{k_{2}}{^\alpha},\ct{m}{^{\alpha}}}$. The strategy is to reach a formula that is completely `symmetric' with respect to $\cbrkt{\ct{k_1}{^\alpha},\ct{k_{2}}{^\alpha},\ct{m}{^{\alpha}}}$ and $\cbrkt{\ct{k_3}{^\alpha},\ct{k_{4}}{^\alpha},\ct{m}{^{\alpha}}}$. Thus, one has to start from the central tetrad but now of \cref{fig:set2}. Some attention must be paid to notation now, as same letters for the components of $\ct{u}{^\alpha}$ in different bases are used, only that the decoration bar `$\ubar{\ }$' is used to distinguish them  from quantities defined in \cref{fig:set1}. Note also that the complex null rotation parameters are called now $f$ (with different decorations), instead of $c$. The proportionality relation between the null vectors can be written as
		\begin{align}
		    \ct{\tilde{k}_{1}}{^\alpha} &= \frac{1}{\alpha_{1}}\ct{k_{1}}{^\alpha}=\tilde{f}\bar{\tilde{f}}\ct{k_1'}{^\alpha}\ ,\label{eq:k1}\\
		     \ct{\breve{k}_{2}}{^\alpha} &= \frac{1}{\alpha_{2}}\ct{k_{2}}{^\alpha}=\breve{f}\bar{\breve{f}}\ct{\hat{k}_{2}}{^\alpha}\ , \label{eq:k2}\\
		     \ct{k_{3}}{^{\alpha}}&=\alpha_{3}\ct{\tilde{k}_{3}}{^{\alpha}}=\alpha_{3}\tilde{c}\bar{\tilde{c}}\ct{k'_{3}}{^{\alpha}}\ ,\label{eq:k3}\\
		     \ct{k_{4}}{^{\alpha}}&=\alpha_{4}\ct{\breve{k}_{4}}{^\alpha}=\alpha_{4}\breve{c}\bar{\breve{c}}\ct{\hat{k}_{4}}{^\alpha}\ .\label{eq:k4}
		\end{align}
	The proportionality functions $\alpha_{i}>0$ obey
		\begin{equation}\label{eq:alphas}
			\alpha_{1}=\breve{c}\bar{\breve{c}}\alpha_{4}\ ,\quad\alpha_{2}=\alpha_{4}\ ,\quad\alpha_{2}=\alpha_{3}\hat{f}\bar{\hat{f}}\ ,\quad\alpha_{1}=\alpha_{3}f'\bar{f}'\tilde{c}\bar{\tilde{c}}\ ,
		\end{equation}
	and		
		\begin{equation}\label{eq:alphas-cs}
			\frac{1}{\alpha_{3}\alpha_{4}}=\frac{1}{\alpha_{2}\alpha_{3}}=-\ct{\breve{k}_{4}}{^\alpha}\ct{\tilde{k}_{3}}{_{\alpha}}=\prn{\tilde{c}-\breve{c}}\prn{\bar{\tilde{c}}-\bar{\breve{c}}}\ .
		\end{equation}
	Some relations between components in the bases of \cref{fig:set1} and \cref{fig:set2} are
		\begin{equation}
			\ubar{b}=\tilde{b}\alpha_{3}\ ,\quad \ubar{a}'=\frac{f'\bar{f}'}{\alpha_{1}}b\ ,\quad \ubar{a}=\alpha_{4}\breve{b}\ ,\quad \ubar{\breve{b}}=\frac{1}{\alpha_{2}}a\ .
		\end{equation}
	Also, note that since $\ct{\breve{\ubar{m}}}{^\alpha}$ ($\hat{\ubar{m}}^\alpha$) and $\ct{\breve{m}}{^{\alpha}}$ ($\ct{\tilde{m}}{^{\alpha}}$) are both orthogonal to $\ct{k_{2}}{^{\alpha}}$, $\ct{k_{4}}{^{\alpha}}$ ($\ct{k_{1}}{^\alpha} $, $\ct{k_{4}}{^{\alpha}}$), $\breve{\ubar{d}}$ ($\hat{\ubar{d}}$) and $\breve{d}$ ($\tilde{d}$) differ at most by a phase, 
		\begin{equation}
			\ubar{\breve{d}}\bar{\ubar{\breve{d}}}=\breve{d}\bar{\breve{d}}\ ,\quad \hat{\ubar{d}}\bar{\hat{\ubar{d}}}=\tilde{d}\bar{\tilde{d}}\ .
		\end{equation}
	Using these relations,
	\begin{equation}\label{eq:J1p}
			\ubar{a}\breve{f}\bar{\breve{f}}\prn{\frac{1}{3}\ubar{\hat{a}}\ubar{\hat{b}}+\ubar{\hat{d}}\bar{\ubar{\hat{d}}}}+\ubar{b}\prn{\frac{1}{3}\ubar{\breve{a}}\ubar{\breve{b}}+\ubar{\breve{d}}\bar{\ubar{\breve{d}}}}+\ubar{\breve{b}}\prn{\frac{1}{3}\ubar{a}\ubar{b}+\ubar{d}\bar{\ubar{d}}}=-\alpha_{3}J_1\ .\\
	\end{equation}
	Similarly, 
	\begin{equation}\label{eq:J2p}
		\ubar{a}\tilde{f}\bar{\tilde{f}}\prn{\frac{1}{3}\ubar{a}'\ubar{b}'+\ubar{d}'\bar{\ubar{d}}'}+\ubar{b}\prn{\frac{1}{3}\tilde{\ubar{a}}\tilde{\ubar{b}}+\tilde{\ubar{d}}\bar{\tilde{\ubar{d}}}}+\tilde{\ubar{b}}\prn{\frac{1}{3}\ubar{a}\ubar{b}+\ubar{d}\bar{\ubar{d}}}=-\hat{c}\bar{\hat{c}}\alpha_{3}J_{2}\ .
	\end{equation}
	But observe that these two equations show that $-J_1$ and $-J_2$ are positive too. Moreover, using \cref{eq:k1,eq:k2,eq:k3,eq:k4} together with \cref{eq:alphas,eq:alphas-cs} and the definition of $\phi_{3}$,
	\begin{equation}
		\alpha_{3}\ubar{\phi}_{3}\bar{\ubar{\phi}_{3}}=\phi_{3}\bar{\phi}_{3}
	\end{equation}
	Taking into account this last relation, inserting \cref{eq:J3p,eq:J4p,eq:J1p,eq:J2p} into \cref{eq:rearraengement} leads to
	\begin{align}
 		\ct{n}{_{\mu}}\ct{\P}{^\mu}\prn{\vec{u}}=4\ubar{\phi}_3\bar{\ubar{\phi}_3}&\Bigg\lbrace B_1\brkt{\ubar{a}\breve{f}\bar{\breve{f}}\prn{\frac{1}{3}\hat{\ubar{a}}\hat{\ubar{b}}+\hat{\ubar{d}}\bar{\hat{\ubar{d}}}}+\ubar{b}\prn{\frac{1}{3}\breve{\ubar{a}}\breve{\ubar{b}}+\breve{\ubar{d}}\bar{\breve{\ubar{d}}}}+\ubar{\breve{b}}\prn{\frac{1}{3}{\ubar{a}}{\ubar{b}}+{\ubar{d}}\bar{{\ubar{d}}}}}\nonumber\\
 		&+B_2\breve{c}\bar{\breve{c}}\brkt{\ubar{a}\tilde{f}\bar{\tilde{f}}\prn{\frac{1}{3}{\ubar{a}^\prime}{\ubar{b}^\prime}+{\ubar{d^\prime}}\bar{{\ubar{d}^\prime}}}+\ubar{b}\prn{\frac{1}{3}\tilde{\ubar{a}}\tilde{\ubar{b}}+\tilde{\ubar{d}}\bar{\tilde{\ubar{d}}}}+\ubar{\tilde{b}}\prn{\frac{1}{3}{\ubar{a}}{\ubar{b}}+{\ubar{d}}\bar{{\ubar{d}}}}}\Bigg\rbrace\nonumber\\
 		+4\phi_{3}\bar{\phi_{3}}&\Bigg\lbrace\tilde{B}_{3}\brkt{{a}\breve{c}\bar{\breve{c}}\prn{\frac{1}{3}\hat{{a}}\hat{{b}}+\hat{{d}}\bar{\hat{{d}}}}+{b}\prn{\frac{1}{3}\breve{{a}}\breve{{b}}+\breve{{d}}\bar{\breve{{d}}}}+{\breve{b}}\prn{\frac{1}{3}{{a}}{{b}}+{{d}}\bar{{{d}}}}}\nonumber\\
 		&+\breve{B}_4\brkt{{a}\tilde{c}\bar{\tilde{c}}\prn{\frac{1}{3}{{a}^\prime}{{b}^\prime}+{{d^\prime}}\bar{{{d}^\prime}}}+{b}\prn{\frac{1}{3}\tilde{{a}}\tilde{{b}}+\tilde{{d}}\bar{\tilde{{d}}}}+{\tilde{b}}\prn{\frac{1}{3}{{a}}{{b}}+{{d}}\bar{{{d}}}}}\Bigg\rbrace\ .\label{eq:general-formula}
 	\end{align}
	 	This is the formula reported in the main text. It serves to show at a glance that if all PNDs have the same orientation with respect to the unit spacelike vector $\ct{n}{^{\alpha}}$ (e.g., $B_i>0$ or $B_i<0$ for all $i$), $\ct{n}{_{\mu}}\ct{\P}{^\mu}\prn{\vec{u}}$ inherits that sign, and if all PNDs are orthogonal to $\ct{n}{^{\alpha}}$, then $\ct{n}{_{\mu}}\ct{\P}{^\mu}\prn{\vec{u}}$ vanishes. This formula is valid for any Weyl tensor candidate with any Petrov type. In the following, we consider the algebraically special cases, which lead to simplifications of \cref{eq:general-formula}.
	 	
	 		\subsection*{Type II}
	 		To find \cref{eq:general-formula} in the algebraically special cases, one can start from scratch, following the same steps indicated before, and using that now one of the PND's is repeated. For that, indeed, there is no need to use both sets of bases shown in \cref{fig:set1,fig:set2}. Instead, what will be shown next is how to make the limit from \cref{eq:general-formula} to algebraically special cases.\\
	 		
	 		To do so, it is more easy to start one step before, with \cref{eq:rearraengement} with the $J_i$ given in one set of bases \cref{eq:J1,eq:J2,eq:J3,eq:J4}. First, observe that according to \cref{fig:set1}, to have a repeated PND one needs to set $\ct{k_{1}}{^\alpha}=\ct{\tilde{k}_{3}}{^{\alpha}}$. This means
	 			\begin{equation}
	 				\tilde{c}=0\ .
	 			\end{equation}
	 		One then has
	 			\begin{equation}
	 			-J_1=-2a\breve{c}\bar{\breve{c}}\prn{ab+d\bar{d}}+\prn{3ab+d\bar{d}}\prn{\breve{c}\bar{d}+\bar{\breve{c}}d}-2b\prn{ab+d\bar{d}}\ .
	 			\end{equation}
	 		However, this has a manifestly positive expression,
	 			\begin{equation}\label{eq:J1-II}
	 			-J_1=a\breve{c}\bar{\breve{c}}\prn{\frac{1}{3}\hat{a}\hat{b}+\hat{d}\bar{\hat{d}}}+b\prn{\frac{1}{3}\breve{a}\breve{b}+\breve{d}\bar{\breve{d}}}+\breve{b}\prn{\frac{1}{3}ab+d\bar{d}}.	
	 			\end{equation}
	 		This should not be a surprise, as this is the expression for $J_3$ in \cref{eq:J3p}. Also
	 			\begin{equation}
	 				B_1=\tilde{B}_{3}\ ,
	 			\end{equation}
	 		so that in the final formula one will sum both contributions (\cref{eq:J1-II,eq:J3p}) leaving on single term multiplied by $B_1$.
	 		Setting $\tilde{c}=0$ now in \cref{eq:J2} gives
	 			\begin{equation}
	 				-J_2=2b\breve{c}\bar{\breve{c}}\prn{ab+d\bar{d}}-4b^2\prn{\breve{c}\bar{d}+\bar{\breve{c}}d}+4b^3\ ,
	 			\end{equation}
	 		which, recalling \cref{eq:c-inverse}, gives this time
	 			\begin{equation}
	 				-J_2=2b\breve{c}\bar{\breve{c}}\prn{\hat{a}\hat{b}+\hat{d}\bar{\hat{d}}}\ .
	 			\end{equation}
	 		It only remains to evaluate $J_{4}$. This is easy, as from \cref{eq:J4} it reads
	 			\begin{equation}
	 				-J_4=2b\prn{ab+d\bar{d}}\ .
	 			\end{equation}
	 		Finally, and to make the subsequent limit to type D easier, use \cref{eq:phis-relation-1} to write
	 			\begin{equation}
	 				\phi_{3}\bar{\phi}_{3}=\frac{9}{4}\hat{c}\bar{\hat{c}}\phi_{2}\bar{\phi}_{2}\ .
	 			\end{equation}
	 		The final formula for type II reads
	 			\begin{align}	\ct{n}{_{\mu}}\ct{\P}{^\mu}\prn{\vec{u}}=18{\phi}_2\bar{{\phi}_2}&\Bigg\lbrace B_1\brkt{a\prn{\frac{1}{3}\hat{a}\hat{b}+\hat{d}\bar{\hat{d}}}+b\hat{c}\bar{\hat{c}}\prn{\frac{1}{3}\breve{a}\breve{b}+\breve{d}\bar{\breve{d}}}+\hat{a}\prn{\frac{1}{3}ab+d\bar{d}}}\nonumber\\
			 		&+B_2b\prn{\hat{a}\hat{b}+\hat{d}\bar{\hat{d}}}+\breve{B}_4\hat{c}\bar{\hat{c}}{b\prn{ab+d\bar{d}}}\Bigg\rbrace\ .\label{eq:typeII}
	 			\end{align}
	 		\subsection*{Type D}
	 			To derive the formula for type D from case type II, one has to make the limit $\hat{c}=0$. First recall \cref{eq:B4,eq:c-inverse}, which tells that $\hat{c}\bar{\hat{c}}\breve{B}_4=\hat{B}_{4}$ is safe in the limit. Also, one can compute separately the limit of the following term
	 				\begin{align}
	 					\lim_{\hat{c}\to 0} b\hat{c}\bar{\hat{c}}\prn{\frac{1}{3}\breve{a}\bar{\breve{b}}+\breve{d}\bar{\breve{d}}}=\frac{4}{3}a^2b\ ,
	 				\end{align}
	 			which follows from the `$\breve{\ }$' version of \cref{eq:tildea,eq:tildeb,eq:tilded}. Observe that in this case $\hat{a}=a$ and $\hat{d}=d$ and that $\hat{B}_{4}=B_{2} $. All in all, one gets
	 				\begin{equation}\label{eq:typeD}
	 				\ct{n}{_{\mu}}\ct{\P}{^\mu}\prn{\vec{u}}=36{\phi}_2\bar{{\phi}_2}\prn{aB_{1}+bB_{2}}\prn{ab+d\bar{d}}\ .
	 				\end{equation}
	 		\subsection*{Type III}
	 			For this, $\breve{c}=0$, and it is more convenient to start with \cref{eq:typeII} written in terms of $\phi_{3}$, i.e.,
 					\begin{align}	\ct{n}{_{\mu}}\ct{\P}{^\mu}\prn{\vec{u}}=8{\phi}_3\bar{{\phi}_3}&\Bigg\lbrace B_1\brkt{a\breve{c}\bar{\breve{c}}\prn{\frac{1}{3}\hat{a}\hat{b}+\hat{d}\bar{\hat{d}}}+b\prn{\frac{1}{3}\breve{a}\breve{b}+\breve{d}\bar{\breve{d}}}+\breve{b}\prn{\frac{1}{3}ab+d\bar{d}}}\nonumber\\
 							 		&+B_2b\breve{c}\bar{\breve{c}}\prn{\hat{a}\hat{b}+\hat{d}\bar{\hat{d}}}+\breve{B}_4{b\prn{ab+d\bar{d}}}\Bigg\rbrace\ .
 					 			\end{align}
 				Next, observe that $\breve{B}_{4}=B_{1}$, and $\breve{b}=b$, $\breve{d}=d$, $\breve{a}=a$. Also, using the null-rotation transformation formulae,
 				\begin{align}
 				\lim_{\breve{c}\to 0}\breve{c}\bar{\breve{c}}\hat{a}\hat{b}=\lim_{\breve{c}\to 0}\breve{c}\bar{\breve{c}}\hat{d}\bar{\hat{d}}=2b^2\ .
 				\end{align}
 			Using all these relations, one gets
 				\begin{equation}
 					\ct{n}{_{\mu}}\ct{\P}{^\mu}\prn{\vec{u}}=8{\phi}_3\bar{{\phi}_3}\brkt{ B_13b\prn{ab+d\bar{d}}+b^3B_2} .
 				\end{equation}
	 		\subsection*{Type N}
		 		Finally, type N cannot be derived by taking the limit, since in the bases \cref{fig:set1,fig:set2} one assumes there are at least two different PNDs. In the case considered now, there is just one repeated PND $\ct{k_{1}}{^\alpha}$. Hence, one has to derive the result directly from \cref{eq:nP-original}. Luckily, this is a trivial task, as now all $\phi_i$ vanish except for $\phi_4\neq 0$. This results into
	 			\begin{equation}
	 					\ct{n}{_{\mu}}\ct{\P}{^\mu}\prn{\vec{u}}=4\phi_{4}\bar{\phi}_{4}b^3B_{1}\ .
	 			\end{equation}

\printbibliography
\end{document}